\documentclass[a4paper]{llncs}

\usepackage{amsfonts}
\usepackage{graphicx}
\usepackage{subfigure}
\usepackage{footnote}
\usepackage{arydshln}

\usepackage{enumerate}

\RequirePackage{graphics}
\usepackage{xcolor}
\usepackage{footnote}
\usepackage{xparse}
\usepackage{cite}
\usepackage{booktabs}

\usepackage[section]{algorithm}
\usepackage{algorithmicx}
\usepackage[noend]{algpseudocode}

\usepackage{amssymb}
\usepackage{amsmath}
\usepackage{keyval}
\usepackage{xspace}
\usepackage{paralist}
\usepackage{listings}
\usepackage{multirow}
\usepackage{stmaryrd}
\usepackage{adjustbox} 
\usepackage{makecell}

\RequirePackage{tikz}
\usetikzlibrary{arrows,automata,shapes,calc,through,decorations.pathmorphing,decorations.fractals,chains,shapes.multipart}

\usepackage{arydshln} 
\usepackage[free-standing-units]{siunitx}

\usepackage{circuitikz}

\def\titlename{Simulation-Based Reachability Analysis for High-Index Large Linear Differential Algebraic Equations\xspace}
\def\authortran{Hoang-Dung Tran}

\def\authorweiming{Weiming Xiang}
\def\authornate{Nathaniel Hamilton}
\def\authortaylor{Taylor T. Johnson}

\usepackage[pdftex]{hyperref}
\hypersetup{
  colorlinks				=true,
  citecolor					={purple},
  linkcolor 				={blue},
  bookmarksnumbered	=true
}

\newcommand{\commenttaylor}[1]{}

\newcommand{\nnnum}[1]{\relax\ifmmode
  {\mathbb #1}_{\geq 0} \else ${\mathbb #1}_{\geq 0}$
  \fi}
\newcommand{\npnum}[1]{\relax\ifmmode
  {\mathbb #1}_{\leq 0} \else ${\mathbb #1}_{\leq 0}$
  \fi}
\newcommand{\pnum}[1]{\relax\ifmmode
  {\mathbb #1}_{> 0} \else ${\mathbb #1}_{> 0}$
  \fi}
\newcommand{\nnum}[1]{\relax\ifmmode
  {\mathbb #1}_{< 0} \else ${\mathbb #1}_{< 0}$
  \fi}
\newcommand{\plnum}[1]{\relax\ifmmode
  {\mathbb #1}_{+} \else ${\mathbb #1}_{+}$
  \fi}
\newcommand{\nenum}[1]{\relax\ifmmode
  {\mathbb #1}_{-} \else ${\mathbb #1}_{-}$
  \fi}

\newcommand{\extb}[1]{\relax\ifmmode {\sf ExtBeh}_{#1} \else ${\sf ExtBeh}_{#1}$\fi}
\newcommand{\tdists}[1]{\relax\ifmmode {\sf Tdists}_{#1} \else ${\sf Tdists}_{#1}$\fi}

\newcommand{\exec}[1]{\relax\ifmmode {\sf Execs}_{#1} \else ${\sf Exec}_{#1}$\fi}
\newcommand{\execf}[1]{\relax\ifmmode {\sf Execs}^*_{#1} \else ${\sf Exec}^*_{#1}$\fi}
\newcommand{\execi}[1]{\relax\ifmmode {\sf Execs}^\omega_{#1} \else ${\sf Exec}^\omega_{#1}$\fi}

\newcommand{\ctrace}[1]{\relax\ifmmode {\sf Ctraces}_{#1} \else ${\sf Ctraces}_{#1}$\fi}

\newcommand{\trace}[1]{\relax\ifmmode {\sf Traces}_{#1} \else ${\sf Traces}_{#1}$\fi}
\newcommand{\tracef}[1]{\relax\ifmmode {\sf Traces}^*_{#1} \else ${\sf Traces}^*_{#1}$\fi}
\newcommand{\tracei}[1]{\relax\ifmmode {\sf Traces}^\omega_{#1} \else ${\sf Traces}^\omega_{#1}$\fi}

\newcommand{\frag}[1]{\relax\ifmmode {\sf Frags}_{#1} \else ${\sf Frags}_{#1}$\fi}
\newcommand{\fragf}[1]{\relax\ifmmode {\sf Frags}^*_{#1} \else ${\sf Frags}^*_{#1}$\fi}
\newcommand{\fragi}[1]{\relax\ifmmode {\sf Frags}^\omega_{#1} \else ${\sf Frags}^\omega_{#1}$\fi}

\newcommand{\reach}[1]{\relax\ifmmode {\sf Reach}_{#1} \else ${\sf Reach}_{#1}$\fi}

\def\A{{\cal A}} 
\def\E{{\cal E}} 
\def\I{{\cal I}} 
\def\R{{\cal R}} 
\def\T{{\cal T}} 
\def\U{{\cal U}} 
\def\X{{\cal X}} 

\newcommand{\col}[1]{\relax\ifmmode \mathscr #1\else $\mathscr #1$\fi}

\definecolor{HIOAcolor}{rgb}{0.776,0.22,0.07}

\newcommand{\SC}[2]{\relax\ifmmode {\tt Scount}(#1,#2) \else ${\tt Scount}(#1,#2)$\fi}
\newcommand{\SCM}[2]{\relax\ifmmode {\tt Smin}(#1,#2) \else ${\tt Smin}(#1,#2)$\fi}
\newcommand{\Aut}[1]{\relax\ifmmode {\tt Aut}(#1) \else ${\tt Aut}(#1)$\fi}

\newcommand{\act}[1]{{\operatorname{\mathsf{#1}}}}

\newcommand{\seclabel}[1]{\label{sec:#1}}

\renewcommand{\eqref}[1]{Equation~\ref{eq:#1}}

\newcommand{\remove}[1]{}
\newcommand{\salg}[1]{\relax\ifmmode {\mathcal F}_{#1}\else ${\mathcal F}_{#1}$\fi}
\newcommand{\msp}[1]{\relax\ifmmode (#1, \salg{#1}) \else $(#1, \salg{#1})$\fi}
\newcommand{\msprod}[2]{\relax\ifmmode ( #1 \times #2, \salg{#1} \otimes \salg{#2}) \else $(#1 \times #2, \salg{#1} \otimes \salg{#2})$\fi}
\newcommand{\dist}[1]{\relax\ifmmode {\mathcal P}\msp{#1}
  \else ${\mathcal P}\msp{#1}$\fi}
\newcommand{\subdist}[1]{\relax\ifmmode {\mathcal S}{\mathcal P}\msp{#1}
  \else ${\mathcal S}{\mathcal P}\msp{#1}$\fi}
\newcommand{\disc}[1]{\relax\ifmmode {\sf Disc}(#1)
  \else ${\sf Disc}(#1)$\fi}

\newcommand{\Trajeq}{\relax\ifmmode {\mathcal R}_\T \else ${\mathcal R}_\T$\fi}
\newcommand{\Acteq}{\relax\ifmmode {\mathcal R}_A \else ${\mathcal R}_A$\fi}
\newcommand{\noop}{\relax\ifmmode \lambda \else $\lambda$\fi}
\newcommand{\close}[1]{\relax\ifmmode \overline{#1} \else $\overline{#1}$\fi}

\newcommand{\tup}[1]
           {
             \relax\ifmmode
             \langle #1 \rangle
             \else $\langle$ #1 $\rangle$ \fi
           }

\newcommand{\lit}[1]{ \relax\ifmmode
                \mathord{\mathcode`\-="702D\sf #1\mathcode`\-="2200}
                \else {\it #1} \fi }

\newcommand{\figuresize}{\scriptsize}

\lstdefinelanguage{ioa}{
  basicstyle=\figuresize,
  keywordstyle=\bf \figuresize,
  identifierstyle=\it \figuresize,
  emphstyle=\tt \figuresize,
  mathescape=true,
  tabsize=20,
  sensitive=false,
  columns=fullflexible,
  keepspaces=false,
  flexiblecolumns=true,
  basewidth=0.05em,
  escapeinside={(*@}{@*)},
  moredelim=[il][\rm]{//},
  moredelim=[is][\sf \figuresize]{!}{!},
  moredelim=[is][\bf \figuresize]{*}{*},
  keywords={automaton,and,
  	 choose,const,continue, components,
  	 discrete, do,
  	 eff, Eff, external,else, elseif, evolve, end,
  	 fi,for, forward, from,
  	 hidden,
  	 in,input,internal,if,invariant, initially, imports,
     let,
     or, output, operators, od, of,
     pre, Pre,
     return,
     such,satisfies, stop, signature, simulation,
     trajectories,trajdef, transitions, that,then, type, types, to, tasks,
     variables, vocabulary,
     when,where, with,while},
  emph={set, seq, tuple, map, array, enumeration},
   literate=
        {(}{{$($}}1
        {)}{{$)$}}1
        {\\in}{{$\in\ $}}1
        {\\preceq}{{$\preceq\ $}}1
        {\\subset}{{$\subset\ $}}1
        {\\subseteq}{{$\subseteq\ $}}1
        {\\supset}{{$\supset\ $}}1
        {\\supseteq}{{$\supseteq\ $}}1
        {\\forall}{{$\forall$}}1
        {\\le}{{$\le\ $}}1
        {\\ge}{{$\ge\ $}}1
        {\\gets}{{$\gets\ $}}1
        {\\cup}{{$\cup\ $}}1
        {\\cap}{{$\cap\ $}}1
        {\\langle}{{$\langle$}}1
        {\\rangle}{{$\rangle$}}1
        {\\exists}{{$\exists\ $}}1
        {\\bot}{{$\bot$}}1
        {\\rip}{{$\rip$}}1
        {\\emptyset}{{$\emptyset$}}1
        {\\notin}{{$\notin\ $}}1
        {\\not\\exists}{{$\not\exists\ $}}1
        {\\ne}{{$\ne\ $}}1
        {\\to}{{$\to\ $}}1
        {\\implies}{{$\implies\ $}}1
        {<}{{$<\ $}}1
        {>}{{$>\ $}}1
        {=}{{$=\ $}}1
        {~}{{$\neg\ $}}1
        {|}{{$\mid$}}1
        {'}{{$^\prime$}}1
        {\\A}{{$\forall\ $}}1
        {\\E}{{$\exists\ $}}1
        {\\nE}{{$\nexists\ $}}1
        {\\/}{{$\vee\,$}}1
        {\\vee}{{$\vee\,$}}1
        {/\\}{{$\wedge\,$}}1
        {\\wedge}{{$\wedge\,$}}1
        {=>}{{$\Rightarrow\ $}}1
        {->}{{$\rightarrow\ $}}1
        {<=}{{$\Leftarrow\ $}}1
        {<-}{{$\leftarrow\ $}}1
        {~=}{{$\neq\ $}}1
        {\\U}{{$\cup\ $}}1
        {\\I}{{$\cap\ $}}1
        {|-}{{$\vdash\ $}}1
        {-|}{{$\dashv\ $}}1
        {<<}{{$\ll\ $}}2
        {>>}{{$\gg\ $}}2
        {||}{{$\|$}}1
        {[}{{$[$}}1
        {]}{{$\,]$}}1
        {[[}{{$\langle$}}1
        {]]]}{{$]\rangle$}}1
        {]]}{{$\rangle$}}1
        {<=>}{{$\Leftrightarrow\ $}}2
        {<->}{{$\leftrightarrow\ $}}2
        {(+)}{{$\oplus\ $}}1
        {(-)}{{$\ominus\ $}}1
        {_i}{{$_{i}$}}1
        {_j}{{$_{j}$}}1
        {_{i,j}}{{$_{i,j}$}}3
        {_{j,i}}{{$_{j,i}$}}3
        {_0}{{$_0$}}1
        {_1}{{$_1$}}1
        {_2}{{$_2$}}1
        {_n}{{$_n$}}1
        {_p}{{$_p$}}1
        {_k}{{$_n$}}1
        {-}{{$\ms{-}$}}1
        {@}{{}}0
        {\\delta}{{$\delta$}}1
        {\\R}{{$\R$}}1
        {\\Rplus}{{$\Rplus$}}1
        {\\N}{{$\N$}}1
        {\\times}{{$\times\ $}}1
        {\\tau}{{$\tau$}}1
        {\\alpha}{{$\alpha$}}1
        {\\beta}{{$\beta$}}1
        {\\gamma}{{$\gamma$}}1
        {\\ell}{{$\ell\ $}}1
        {--}{{$-\ $}}1
        {\\TT}{{\hspace{1.5em}}}3
      }

\lstdefinelanguage{ioaNums}[]{ioa}
{
  numbers=left,
  numberstyle=\tiny,
  stepnumber=2,
  numbersep=4pt
}

\lstdefinelanguage{ioaNumsRight}[]{ioa}
{
  numbers=right,
  numberstyle=\tiny,
  stepnumber=2,
  numbersep=4pt
}

\lstnewenvironment{IOA}%
  {\lstset{language=IOA}}
  {}

\lstnewenvironment{IOANums}%
  {
  \if@firstcolumn
    \lstset{language=IOA, numbers=left, firstnumber=auto}
  \else
    \lstset{language=IOA, numbers=right, firstnumber=auto}
  \fi
  }
  {}

\lstnewenvironment{IOANumsRight}%
  {
    \lstset{language=IOA, numbers=right, firstnumber=auto}
  }
  {}

\newcommand{\linefigioa}[9]{

}

\lstdefinelanguage{ioaLang}{%
  basicstyle=\ttfamily\small,
  keywordstyle=\rmfamily\bfseries\small,
  identifierstyle=\small,
  keywords={assumes,automaton,axioms,backward,bounds,by,case,choose,components,const,d,det,discrete,do,eff,else,elseif,ensuring,enumeration,evolve,fi,fire,follow,for,forward,from,hidden,if,in,%
    input,initially,internal,invariant,let, local,od,of,output,pre,schedule,signature,so,%
    simulation,states,variables, tasks, stop,tasks,that,then,to,trajdef,trajectory,trajectories,transitions,tuple,type,union,urgent,uses,when,where,while,yield},
  literate=
        {\\in}{{$\in$}}1
        {\\preceq}{{$\preceq$}}1
        {\\subset}{{$\subset$}}1
        {\\subseteq}{{$\subseteq$}}1
        {\\supset}{{$\supset$}}1
        {\\supseteq}{{$\supseteq$}}1
        {\\rho}{{$\rho$}}1
        {\\infty}{{$\infty$}}1
        {<}{{$<$}}1
        {>}{{$>$}}1
        {=}{{$=$}}1
        {~}{{$\neg$}}1
        {|}{{$\mid$}}1
        {'}{{$^\prime$}}1
        {\\A}{{$\forall$}}1 {\\E}{{$\exists$}}1
        {\\/}{{$\vee$}}1 {/\\}{{$\wedge$}}1
        {=>}{{$\Rightarrow$}}1
        {->}{{$\rightarrow$}}1
        {<=}{{$\leq$}}1 {>=}{{$\geq$}}1 {~=}{{$\neq$}}1
        {\\U}{{$\cup$}}1 {\\I}{{$\cap$}}1
        {|-}{{$\vdash$}}1 {-|}{{$\dashv$}}1
        {<<}{{$\ll$}}2 {>>}{{$\gg$}}2
        {||}{{$\|$}}1
        {<=>}{{$\Leftrightarrow$}}2
        {<->}{{$\leftrightarrow$}}2
        {(+)}{{$\oplus$}}1
        {(-)}{{$\ominus$}}1
}

\lstdefinelanguage{bigIOALang}{%
  basicstyle=\ttfamily,
  keywordstyle=\rmfamily\bfseries,
  identifierstyle=,
  keywords={assumes,automaton,axioms,backward,by,case,choose,components,const,%
    d,det,discrete,do,eff,else,elseif,ensuring,enumeration,evolve,fi,for,forward,from,hidden,if,in%
    input,initially,internal,invariant,local,od,of,output,pre,schedule,signature,so,%
    tasks, simulation,states,stop,tasks,that,then,to,trajdef,trajectories,transitions,tuple,type,union,urgent,uses,when,where,yield},
  literate=
        {\\in}{{$\in$}}1
        {\\preceq}{{$\preceq$}}1
        {\\subset}{{$\subset$}}1
        {\\subseteq}{{$\subseteq$}}1
        {\\supset}{{$\supset$}}1
        {\\supseteq}{{$\supseteq$}}1
        {<}{{$<$}}1
        {>}{{$>$}}1
        {=}{{$=$}}1
        {~}{{$\neg$}}1
        {|}{{$\mid$}}1
        {'}{{$^\prime$}}1
        {\\A}{{$\forall$}}1 {\\E}{{$\exists$}}1
        {\\/}{{$\vee$}}1 {/\\}{{$\wedge$}}1
        {=>}{{$\Rightarrow$}}1
        {->}{{$\rightarrow$}}1
        {<=}{{$\leq$}}1 {>=}{{$\geq$}}1 {~=}{{$\neq$}}1
        {\\U}{{$\cup$}}1 {\\I}{{$\cap$}}1
        {|-}{{$\vdash$}}1 {-|}{{$\dashv$}}1
        {<<}{{$\ll$}}2 {>>}{{$\gg$}}2
        {||}{{$\|$}}1
        {<=>}{{$\Leftrightarrow$}}2
        {<->}{{$\leftrightarrow$}}2
        {(+)}{{$\oplus$}}1
        {(-)}{{$\ominus$}}1
}

\lstnewenvironment{BigIOA}%
  {\lstset{language=bigIOALang,basicstyle=\ttfamily}
   \csname lst@SetFirstLabel\endcsname}
  {\csname lst@SaveFirstLabel\endcsname\vspace{-4pt}\noindent}

\lstnewenvironment{SmallIOA}%
  {\lstset{language=ioaLang,basicstyle=\ttfamily\scriptsize}
   \csname lst@SetFirstLabel\endcsname}
  {\csname lst@SaveFirstLabel\endcsname\noindent}

\newlength{\bracklen}

\newcommand{\tri}[3]{\ensuremath{\mathit{#1}^\mathit{#2}_\mathit{#3}}}

\newcommand{\sugLocalVars}[2]{\ifthenelse{\equal{}{#2}}%
                             {\tri{localVars}{#1}{desug}}%
                             {\tri{localVars}{#1}{#2,desug}}}
\newcommand{\sugVars}[2]{\ifthenelse{\equal{}{#2}}%
                        {\tri{vars}{#1}{desug}}%
                        {\tri{vars}{#1}{#2,desug}}}

\newenvironment{subSyntax}{\begin{array}{l}}{\end{array}}

\newcommand{\ms}[1]{\ifmmode%
\mathord{\mathcode`-="702D\it #1\mathcode`\-="2200}\else%
$\mathord{\mathcode`-="702D\it #1\mathcode`\-="2200}$\fi}

\def\A{{\cal A}} 
\def\T{{\cal T}} 

\lstdefinelanguage{pvs}{
  basicstyle=\tt \figuresize,
  keywordstyle=\sc \figuresize,
  identifierstyle=\it \figuresize,
  emphstyle=\tt \figuresize,
  mathescape=true,
  tabsize=20,
  sensitive=false,
  columns=fullflexible,
  keepspaces=false,
  flexiblecolumns=true,
  basewidth=0.05em,
  moredelim=[il][\rm]{//},
  moredelim=[is][\sf \figuresize]{!}{!},
  moredelim=[is][\bf \figuresize]{*}{*},
  keywords={and,
  	 begin,
  	 cases, const,
  	 do,
  	 external, else, exists, end, endcases, endif,
  	 fi,for, forall, from,
  	 hidden,
  	 in, if, importing,
     let, lambda, lemma,
     measure,
     not,
     or, of,
     return, recursive,
     stop,
     theory, that,then, type, types, type+, to, theorem,
     var,
     with,while},
  emph={nat, setof, sequence, eq, tuple, map, array, enumeration, bool, real, exp, nnreal, posreal},
   literate=
        {(}{{$($}}1
        {)}{{$)$}}1
        {\\in}{{$\in\ $}}1
        {\\mapsto}{{$\rightarrow\ $}}1
        {\\preceq}{{$\preceq\ $}}1
        {\\subset}{{$\subset\ $}}1
        {\\subseteq}{{$\subseteq\ $}}1
        {\\supset}{{$\supset\ $}}1
        {\\supseteq}{{$\supseteq\ $}}1
        {\\forall}{{$\forall$}}1
        {\\le}{{$\le\ $}}1
        {\\ge}{{$\ge\ $}}1
        {\\gets}{{$\gets\ $}}1
        {\\cup}{{$\cup\ $}}1
        {\\cap}{{$\cap\ $}}1
        {\\langle}{{$\langle$}}1
        {\\rangle}{{$\rangle$}}1
        {\\exists}{{$\exists\ $}}1
        {\\bot}{{$\bot$}}1
        {\\rip}{{$\rip$}}1
        {\\emptyset}{{$\emptyset$}}1
        {\\notin}{{$\notin\ $}}1
        {\\not\\exists}{{$\not\exists\ $}}1
        {\\ne}{{$\ne\ $}}1
        {\\to}{{$\to\ $}}1
        {\\implies}{{$\implies\ $}}1
        {<}{{$<\ $}}1
        {>}{{$>\ $}}1
        {=}{{$=\ $}}1
        {~}{{$\neg\ $}}1
        {|}{{$\mid$}}1
        {'}{{$^\prime$}}1
        {\\A}{{$\forall\ $}}1
        {\\E}{{$\exists\ $}}1
        {\\/}{{$\vee\,$}}1
        {\\vee}{{$\vee\,$}}1
        {/\\}{{$\wedge\,$}}1
        {\\wedge}{{$\wedge\,$}}1
        {->}{{$\rightarrow\ $}}1
        {=>}{{$\Rightarrow\ $}}1
        {->}{{$\rightarrow\ $}}1
        {<=}{{$\Leftarrow\ $}}1
        {<-}{{$\leftarrow\ $}}1
        {~=}{{$\neq\ $}}1
        {\\U}{{$\cup\ $}}1
        {\\I}{{$\cap\ $}}1
        {|-}{{$\vdash\ $}}1
        {-|}{{$\dashv\ $}}1
        {<<}{{$\ll\ $}}2
        {>>}{{$\gg\ $}}2
        {||}{{$\|$}}1
        {[}{{$[$}}1
        {]}{{$\,]$}}1
        {[[}{{$\langle$}}1
        {]]]}{{$]\rangle$}}1
        {]]}{{$\rangle$}}1
        {<=>}{{$\Leftrightarrow\ $}}2
        {<->}{{$\leftrightarrow\ $}}2
        {(+)}{{$\oplus\ $}}1
        {(-)}{{$\ominus\ $}}1
        {_i}{{$_{i}$}}1
        {_j}{{$_{j}$}}1
        {_{i,j}}{{$_{i,j}$}}3
        {_{j,i}}{{$_{j,i}$}}3
        {_0}{{$_0$}}1
        {_1}{{$_1$}}1
        {_2}{{$_2$}}1
        {_n}{{$_n$}}1
        {_p}{{$_p$}}1
        {_k}{{$_n$}}1
        {-}{{$\ms{-}$}}1
        {@}{{}}0
        {\\delta}{{$\delta$}}1
        {\\R}{{$\R$}}1
        {\\Rplus}{{$\Rplus$}}1
        {\\N}{{$\N$}}1
        {\\times}{{$\times\ $}}1
        {\\tau}{{$\tau$}}1
        {\\alpha}{{$\alpha$}}1
        {\\beta}{{$\beta$}}1
        {\\gamma}{{$\gamma$}}1
        {\\ell}{{$\ell\ $}}1
        {--}{{$-\ $}}1
        {\\TT}{{\hspace{1.5em}}}3
      }

\lstdefinelanguage{BigPVS}{
  basicstyle=\tt,
  keywordstyle=\sc,
  identifierstyle=\it,
  emphstyle=\tt ,
  mathescape=true,
  tabsize=20,
  sensitive=false,
  columns=fullflexible,
  keepspaces=false,
  flexiblecolumns=true,
  basewidth=0.05em,
  moredelim=[il][\rm]{//},
  moredelim=[is][\sf \figuresize]{!}{!},
  moredelim=[is][\bf \figuresize]{*}{*},
  keywords={and,
  	 begin,
  	 cases, const,
  	 do, datatype,
  	 external, else, exists, end, endif, endcases,
  	 fi,for, forall, from,
  	 hidden,
  	 in, if, importing,
     let, lambda, lemma,
     measure,
     not,
     or, of,
     return, recursive,
     stop,
     theory, that,then, type, types, type+, to, theorem,
     var,
     with,while},
  emph={nat, setof, sequence, eq, tuple, map, array, first, rest, add, enumeration, bool, real, posreal, nnreal},
   literate=
        {(}{{$($}}1
        {)}{{$)$}}1
        {\\in}{{$\in\ $}}1
        {\\mapsto}{{$\rightarrow\ $}}1
        {\\preceq}{{$\preceq\ $}}1
        {\\subset}{{$\subset\ $}}1
        {\\subseteq}{{$\subseteq\ $}}1
        {\\supset}{{$\supset\ $}}1
        {\\supseteq}{{$\supseteq\ $}}1
        {\\forall}{{$\forall$}}1
        {\\le}{{$\le\ $}}1
        {\\ge}{{$\ge\ $}}1
        {\\gets}{{$\gets\ $}}1
        {\\cup}{{$\cup\ $}}1
        {\\cap}{{$\cap\ $}}1
        {\\langle}{{$\langle$}}1
        {\\rangle}{{$\rangle$}}1
        {\\exists}{{$\exists\ $}}1
        {\\bot}{{$\bot$}}1
        {\\rip}{{$\rip$}}1
        {\\emptyset}{{$\emptyset$}}1
        {\\notin}{{$\notin\ $}}1
        {\\not\\exists}{{$\not\exists\ $}}1
        {\\ne}{{$\ne\ $}}1
        {\\to}{{$\to\ $}}1
        {\\implies}{{$\implies\ $}}1
        {<}{{$<\ $}}1
        {>}{{$>\ $}}1
        {=}{{$=\ $}}1
        {~}{{$\neg\ $}}1
        {|}{{$\mid$}}1
        {'}{{$^\prime$}}1
        {\\A}{{$\forall\ $}}1
        {\\E}{{$\exists\ $}}1
        {\\/}{{$\vee\,$}}1
        {\\vee}{{$\vee\,$}}1
        {/\\}{{$\wedge\,$}}1
        {\\wedge}{{$\wedge\,$}}1
        {->}{{$\rightarrow\ $}}1
        {=>}{{$\Rightarrow\ $}}1
        {->}{{$\rightarrow\ $}}1
        {<=}{{$\Leftarrow\ $}}1
        {<-}{{$\leftarrow\ $}}1
        {~=}{{$\neq\ $}}1
        {\\U}{{$\cup\ $}}1
        {\\I}{{$\cap\ $}}1
        {|-}{{$\vdash\ $}}1
        {-|}{{$\dashv\ $}}1
        {<<}{{$\ll\ $}}2
        {>>}{{$\gg\ $}}2
        {||}{{$\|$}}1
        {[}{{$[$}}1
        {]}{{$\,]$}}1
        {[[}{{$\langle$}}1
        {]]]}{{$]\rangle$}}1
        {]]}{{$\rangle$}}1
        {<=>}{{$\Leftrightarrow\ $}}2
        {<->}{{$\leftrightarrow\ $}}2
        {(+)}{{$\oplus\ $}}1
        {(-)}{{$\ominus\ $}}1
        {_i}{{$_{i}$}}1
        {_j}{{$_{j}$}}1
        {_{i,j}}{{$_{i,j}$}}3
        {_{j,i}}{{$_{j,i}$}}3
        {_0}{{$_0$}}1
        {_1}{{$_1$}}1
        {_2}{{$_2$}}1
        {_n}{{$_n$}}1
        {_p}{{$_p$}}1
        {_k}{{$_n$}}1
        {-}{{$\ms{-}$}}1
        {@}{{}}0
        {\\delta}{{$\delta$}}1
        {\\R}{{$\R$}}1
        {\\Rplus}{{$\Rplus$}}1
        {\\N}{{$\N$}}1
        {\\times}{{$\times\ $}}1
        {\\tau}{{$\tau$}}1
        {\\alpha}{{$\alpha$}}1
        {\\beta}{{$\beta$}}1
        {\\gamma}{{$\gamma$}}1
        {\\ell}{{$\ell\ $}}1
        {--}{{$-\ $}}1
        {\\TT}{{\hspace{1.5em}}}3
      }

\lstdefinelanguage{pvsNums}[]{pvs}
{
  numbers=left,
  numberstyle=\tiny,
  stepnumber=2,
  numbersep=4pt
}

\lstdefinelanguage{pvsNumsRight}[]{pvs}
{
  numbers=right,
  numberstyle=\tiny,
  stepnumber=2,
  numbersep=4pt
}

\lstnewenvironment{BigPVS}%
  {\lstset{language=BigPVS}}
  {}

\lstnewenvironment{PVSNums}%
  {
  \if@firstcolumn
    \lstset{language=pvs, numbers=left, firstnumber=auto}
  \else
    \lstset{language=pvs, numbers=right, firstnumber=auto}
  \fi
  }
  {}

\lstnewenvironment{PVSNumsRight}%
  {
    \lstset{language=pvs, numbers=right, firstnumber=auto}
  }
  {}

\newcommand{\linefigpvs}[9]{

}

\lstdefinelanguage{pvsproof}{
  basicstyle=\tt \figuresize,
  mathescape=true,
  tabsize=4,
  sensitive=false,
  columns=fullflexible,
  keepspaces=false,
  flexiblecolumns=true,
  basewidth=0.05em,
}

\def\N{\act{N}}

\newcommand{\localvar}[2]{{{#1_{#2}}}}

\def\xi{\localvar{x}{i}}

\def\reach{{\sf Reach}}

\def\Xi{\mathit{X_i}}

\begin{document}
%

\title{\titlename}

%
%
%

\author{\authorluan,~\authortran,~and~\authortaylor \\ University of Texas at Arlington}%
\author{\authortran \inst{1}, \authorweiming \inst{1}, \authornate \inst{1} \and \authortaylor \inst{1}}

\institute{Department of Computer Science, University of Texas at Arlington, USA\and University of Toronto, Canada \and Vanderbilt University, USA}
\institute{Department of Computer Science and Engineering, University of Texas at Arlington, TX, USA \\ \email{luanvnguyen@mavs.uta.edu} \and Institute for Software Integrated Systems, Vanderbilt University, TN, USA \\ \email{taylor.johnson@vanderbilt.edu}}

\institute{Vanderbilt University, Nashville, TN, USA \\ \email{trhoangdung@gmail.com, xiangwming@gmail.com, taylor.johnson@gmail.com}}
%
%
\maketitle

\begin{abstract}
Reachability analysis is a fundamental problem for safety verification and falsification of Cyber-Physical Systems (CPS) whose dynamics follow  physical laws usually represented as differential equations. In the last two decades, numerous reachability analysis methods and tools have been proposed for a common class of dynamics in CPS known as ordinary differential equations (ODE). However, there is lack of methods dealing with differential algebraic equations (DAE) which is a more general class of dynamics that is widely used to describe a variety of problems from engineering and science such as multibody mechanics, electrical cicuit design, incompressible fluids, molecular dynamics and chemcial process control. Reachability analysis for DAE systems is more complex than ODE systems, especially for \emph{high-index} DAEs because they contain both a \emph{differential part} (i.e., ODE) and \emph{algebraic constraints (AC)}. In this paper, we extend the recent scalable simulation-based reachability analysis in combination with decoupling techniques for a class of high-index large linear DAEs. In particular, a high-index linear DAE is first decoupled into one ODE and one or several AC subsystems based on the well-known Marz decoupling method ultilizing \emph{admissible projectors}. Then, the \emph{discrete} reachable set of the DAE, represented as a list of \emph{star-sets}, is computed using simulation. Unlike ODE reachability analysis where the initial condition is freely defined by a user, in DAE cases, the consistency of the inititial condition is an essential requirement to guarantee a feasible solution. Therefore, a thorough check for the consistency is invoked before computing the discrete reachable set. Our approach sucessfully verifies (or falsifies) a wide range of practical, high-index linear DAE systems in which the number of state variables varies from several to thousands.
\end{abstract}


\section{Introduction}
\seclabel{intro}

Reachability analysis for continuous and hybrid systems has been an attractive research topic for the last two decades since it is an essential problem for verification of safety-critical CPS. In this context, numerous techniques and tools have been proposed. Reachability analysis using zonotopes \cite{althoff2015introduction, girard2005reachability} and support functions \cite{leguernic2010, frehse2011spaceex} are efficient approaches when dealing with linear, continuous and hybrid systems. For nonlinear, continuous and hybrid systems,  dReal\cite{kong2015dreach} using $\delta-$reachability analysis and Flow$^*$ \cite{chen2013flow} using Taylor model are well-known and efficient approaches.

State-space explosion is the main challenge that limits the applicability of over-approximation based approaches \cite{althoff2015introduction, frehse2011spaceex, chen2013flow} to small and medium scale systems. Therefore, some promising approaches have been proposed recently to deal with large scale systems. For linear cases, the simulation-equivalent reachability analysis \cite{bak2017simulation, duggirala2016parsimonious} utilizing the \emph{generalized star set} as the state-set representation has shown an impressive result by successfully dealing with linear systems up to $10,000$ state variables. In this approach, the \emph{discrete simulation-equivalent} reachable set of a linear ODE system can be computed efficiently using standard ODE solvers by taking advantage of the superposition property. This simulation-based method is practically useful for falsification of large systems. Another method utilizes order-reduction abstraction \cite{tran2017order, han2006reachability} in which a large system can be abstracted to a smaller system with bounded error. The smaller system and the error are then used to construct the reachable set or for checking the safety of the original, large system. For nonlinear cases, C2E2 \cite{duggirala2015c2e2, fan2016automatic} utilizing simulation has shown great improvement on time performance and scalability in comparison with other methods.

Although many methods have been developed for reachability analysis of CPS, most of these methods focus on CPS with ODE dynamics. There is a lack of methodology in analyzing systems with DAE dynamics. To the best of the authors' knowledge, there exists only one reachability analysis approach for nonlinear, semi-explicit index-1 DAE systems \cite{althoff2014tac} using zonotopes as state-set representation that has been proposed recently. Dealing with index-1 DAE is slightly different from dealing with pure ODE because, with a consistent initial condition, a semi-explicit index-1 DAE can be converted to an ODE. As CPS involving high-index DAE dynamics appear extensively in engineering and science such as multi-body mechanics, electrical circuit design, heat and gas transfer, chemical process, atmospheric physics, thermodynamic systems, and water distribution network \cite{byrne1988differential}, there is an urgent need for novel reachability analysis methods and tools that can verify (or falsify) the safety property of such CPS. Solving this challenging problem is the main contribution of this research.

In this paper, we investigate the reachability analysis for large linear DAE systems with index up to $3$, which appear widely in practice. There are a variety of definitions for the index of a linear DAE. However, throughout the paper, we use the concept of \emph{tractability index} proposed in \cite{marz1996canonical} to determine \emph{the index of a linear DAE}. Our approach consists of three main steps: decoupling and consistency checking, reachable set computation, and safety verification or falsification. The novelty of our approach comes from its objective in dealing with high-index DAE which is a popular class of dynamics that has not been addressed in existing literature.

In the first step, we use the Marz decoupling method \cite{marz1996canonical, banagaaya2016index} to decouple a high-index DAE into one ODE subsystem and one or several algebraic constraint (AC) subsystems. The core step in decoupling is constructing a set of admissible projectors which has not previously been discussed deeply in existing literature. In this paper, we propose a novel algorithm that can construct such admissible projectors for a linear DAE system with index up to $3$. Additionally, we define \emph{a consistent space} for the DAE because, unlike ODE reachability analysis where the initial set of states can be freely defined by a user, in order to guarantee a numerical solution for the DAE, the initial state and inputs of the system must be consistent and satisfy certain constraints. It is interesting to emphasized that the decoupling and consistency checking methods used in our approach can be combined with existing over-approximation reachability analysis methods \cite{althoff2015introduction, frehse2011spaceex} to compute an over-approximated reachable set for high-index, linear DAE systems with small and medium dimensions.

The second step in our approach is reachable set computation. Since our main objective is to verify or falsify large linear DAEs, we extend ODE simulation-based reachability analysis to DAEs. In particular, we modify the generalized star-set proposed in \cite{bak2017simulation} to enhance the efficiency in checking the initial condition consistency and safety for DAEs. From a consistent initial set of states and inputs, the reachable set of a DAE can be constructed by combining the reachable sets of its subsystems. It is also worth pointing out that the piecewise constant inputs assumption for ODE with inputs used in \cite{bak2017simulation} may lead a DAE system to \emph{impulsive behavior}. Therefore, in this paper, we assume the inputs applied to the system are \emph{smooth functions}. These kinds of inputs can be obtained by \emph{smoothing} piecewise constant inputs with filters.

The last step in our approach is to verify or falsify the safety property of the DAE system using the constructed reachable set. In this paper, we consider linear safety specifications. We are interested in checking the safety of the system in a specific direction defined using a directional matrix. Using the modified star-set and the directional matrix, checking the safety property can be solved efficiently as a low-dimensional feasibility linear programming problem. In the case of violation, our approach generates a trace that falsifies the system safety.
\section{Preliminaries}
\label{sec:preliminaries}
\subsection{Linear DAE system}
We are interested in the reachability analysis of a high-index large linear DAE system described as follows:
\begin{equation} \label{eq:dae}
\Delta:~E\dot{x}(t) = Ax(t) + Bu(t),
\end{equation}
where $x(t) \in \mathbb{R}^n$ is the state vector of the system; $E, A \in \mathbb{R}^{n \times n}$, $B \in \mathbb{R}^{n \times m}$ are the system's matrices in which $E$ is \emph{singular}; and $u(t) \in \mathbb{R}^{m}$ is the input of the system. Let $I_n$ be the $n$-dimensional identity matrix. The regularity, the tractability index, the admissible projectors, the fixed-step bounded-time simulation, and the bounded-time simulation-equivalent reachable set of the system are defined below.
\begin{definition}[Regularity \cite{dai1989singular}]
The pair $(E, A)$ is said to be regular if $det(sE - A)$ is not identically zero.
\end{definition}
\begin{remark}
For \emph{any specified initial conditions}, the regularity of the pair $(E, A)$ guarantees the existence and uniqueness of a solution of the system (\ref{eq:dae}).
\end{remark}
\begin{definition}[Tractability index \cite{marz1996canonical}]\label{def:index}
Assume that the DAE system (\ref{eq:dae}) is solvable, i.e., the matrix pair $(E, A)$ is regular. A matrix chain is defined by:
\begin{equation}\label{eq:matrix_chain}
\begin{split}
&E_0 = E,~ A_0 = A, \\
&E_{j + 1} = E_{j} - A_jQ_j,~~A_{j + 1} = A_jP_j, ~for~j \geq 0,
\end{split}
\end{equation}
where $Q_j$ are projectors onto $Ker(E_j)$, i.e., $E_jQ_j = 0,~Q_j^2 = Q_j$, and $P_j = I_n - Q_j$. Then, there exists an index $\mu$ such that $E_{\mu}$ is nonsingular and all $E_j$ are singular for $0 \leq j < \mu - 1$. It is said that the system (\ref{eq:dae}) has tractability index-$\mu$. In  the rest of the paper, we use the term ``index'' to state for the ``tractability index'' of the system.
\end{definition}
\begin{definition}[Admissible projectors \cite{marz1996canonical}]\label{def:admissible_projectors}
Given a DAE with tractability index-$\mu$, the projectors $Q_0, Q_1, \cdots, Q_{\mu - 1}$ in Definition \ref{def:index} are called admissible if and only if they satisfy the following property: $\forall j > i $, $Q_jQ_i = 0$.
\end{definition}

\begin{definition}[Fixed-step, bounded-time simulation]\label{def:simulation} Given consistent initial state $x_0$ and input $u(t)$, a time bound $T$, and a time step $h$, the finite sequence:
\vspace{-0.5em}
\begin{equation*}
 \rho(x_0, u(t), h, T=Nh) = x_0 \xrightarrow[0 \leq t < h]{u(t)}x_1 \xrightarrow[h \leq t < 2h]{u(t)}x_2 \cdots \xrightarrow[(N-1)h \leq t < Nh]{u(t)}x_{N},
\end{equation*}
is a $(x_0, u(t), h, T)$-simulation of the DAE system (\ref{eq:dae}) if and only if for all $0 \leq i \leq N-1$, $x_{i + 1}$ is the state of the system trajectory starting from $x_i$ when provided with input function $u(t)$ for $ih \leq t < (i + 1)h$. If there is no input, $u(t) = 0$.
\end{definition}

The consistent condition for the initial state $x_0$ and input $u(t)$ will be discussed in detail in Section \ref{sec:consistency}. From the fixed-step, bounded-time simulation of a DAE system, we define the following bounded-time, simulation-equivalent reachable set of the DAE system.
%
\begin{definition}[Bounded-time, simulation-equivalent reachable set]\label{def:reachable_set}
Given sets of consistent initial state $X_0$ and input $U$, the bounded-time, simulation-equivalent reachable set $R_{[0,T]}(\Delta)$ of the system (\ref{eq:dae}) is the set of all states that can be encountered by any $(x_0, u(t), h, T)$-simulation starting from any $x_0 \in \X_0$ and input $u(t) \in U$.
\end{definition}

Let $Unsafe(\Delta) \triangleq Gx \leq f$ be the unsafe set of the DAE system (\ref{eq:dae}) in which $x \in \mathbb{R}^n$ is the state vector of the system, $G \in \mathbb{R}^{k \times n}$ is the \emph{unsafe matrix} and $f \in \mathbb{R}^{k}$ is the \emph{unsafe vector}. Given sets of consistent initial state $X_0$ and input $U$, the simulation-based safety verification and falsification problem is defined in the following.
\begin{definition}[Simulation-based safety verification and falsification]
The DAE system (\ref{eq:dae}) is said to be ``simulationally safe'' up to time $T$ if and only if its simulation-equivalent reachable set, $R_{[0, T]}(\Delta)$, and the unsafe set, $Unsafe(\Delta)$, are disjoint, i.e., $R_{[0, T]}(\Delta) \cap Unsafe(\Delta) = \emptyset$. Otherwise, it is simulationally unsafe.

The DAE system is said to be ``simulationally falsifiable'' if and only if it is simulationally unsafe and there exists a simulation, $(x_0, u(t), h, T)$, that leads the initial state, $x_0$, of the system to an unsafe state, $x_{unsafe} \in Unsafe(\Delta)$.
\end{definition}

The main objective of the paper is to compute the simulation-equivalent reachable set, $R_{[0,T]}(\Delta)$, of the DAE system and use it to verify or falsify the safety property of the system. In the rest of the paper, we use the term \emph{reachable set} to stand for \emph{simulation-equivalent reachable set}. Next, we define a \emph{modified star set} which is used as the state-set representation of the DAE system.
\vspace{-1em}
\subsection{Modified star set}
In our approach, we use a modified star set to represent the reachable set of the DAE system. The modified star set is slightly different from the generalized star set \cite{bak2017simulation} because it does not have a \emph{center vector} and is only defined on a star's $n \times k$ basis matrix.
\vspace{-0.25em}
\begin{definition} [Modified star set] A modified star set (or simply star) $\Theta$ is a tuple $\langle V, P \rangle$ where $ V = [v_1, v_2, \cdots, v_k] \in \mathbb{R}^{n \times k}$ is a star basis matrix and $P$ is a linear predicate. The set of states represented by the star is given by:
\begin{equation}
 \llbracket \Theta \rrbracket = \{x~|~x = \Sigma_{i=1}^k(\alpha_iv_i) = V \times \alpha,~P(\alpha) \triangleq C\alpha \leq d \},
\end{equation}
where $\alpha = [\alpha_1, \alpha_2, \cdots, \alpha_k]^{\textbf{T}}$, $C \in \mathbb{R}^{p \times k}$, $d \in \mathbb{R}^p$ and $p$ is the number of linear constraints.
\end{definition}

One can see that if we let $v_1$ be a \emph{center vector} and $\alpha_1 = 1$, then the modified star set becomes the generalized star set proposed in \cite{bak2017simulation}. The benefit of the modified star set come from its form given as a \emph{matrix-vector product} which is convenient for checking initial condition consistency and safety properties. In the rest of the paper, we will refer to both the tuple $\Theta$ and the set of states $\llbracket \Theta \rrbracket$ as $\Theta$.

To construct the reachable set of the DAE system (\ref{eq:dae}), we decouple the system into $\mu + 1$ subsystems where $\mu$ is the index of the DAE system. The underlining technique used in our approach is the Marz decoupling method utilizing admissible projectors which is presented in detail in the following section.

\vspace{-1em}
\section{Decoupling}
\label{sec:decoupling}
\vspace{-0.5em}
In this section, we discuss how to decouple a high-index DAE system into one ODE subsystem and one or several AC subsystems using the matrix chain and admissible projectors defined in the previous section. Since we are particularly interested in DAE systems with index up to 3, the proofs of decoupling process for index-1, -2, and -3 are given in detail. A generalization of decoupling for a DAE with arbitrary index is presented in \cite{marz1996canonical}. As the construction of admissible projectors used in decoupling has not been discussed clearly in existing literature, in this section, we propose a method and an algorithm to solve this problem.
\vspace{-0.5em}
\begin{lemma}[Index-1 DAE decoupling~\cite{marz1996canonical, banagaaya2016index}]\label{lm:index-1-decoupling}
An index-1 DAE system described by (\ref{eq:dae}) can be decoupled using the matrix chain defined by Equation (\ref{eq:matrix_chain}) as follows:
\begin{equation*}
\begin{split}
\Delta_1:~~\dot{x}_1(t) &= N_1x_1(t) + M_1u(t),~\text{ODE subsystem}, \\
\Delta_2:~~x_2(t) &= N_2x_1(t) + M_2u(t),~\text{AC subsystem}, \\
~~~~~~~~~x(t) &= x_1(t) + x_2(t), \\
~~~~~~~~~x_1(t) &= P_0x(t),~N_1 = P_0E_1^{-1}A_0,~M_1 = P_0E_1^{-1}B, \\
~~~~~~~~~x_2(t) &= Q_0x(t),~N_2 = Q_0E_1^{-1}A_0,~M_2 = Q_0E_1^{-1}B.
\end{split}
\end{equation*}
\end{lemma}
Proof is given in Appendix \ref{proof:index-1}.
\vspace{-0.5em}
\begin{lemma}[Index-2 DAE decoupling~\cite{marz1996canonical, banagaaya2016index}]\label{lm:index-2-decoupling}
An index-2 DAE system described by (\ref{eq:dae}) can be decoupled into a decoupled system using the matrix chain defined by Equation (\ref{eq:matrix_chain}) and the admissible projectors in Definition \ref{def:admissible_projectors} as follows:
\begin{equation*}
\begin{split}
\Delta_1:~~\dot{x}_1(t) &= N_1x_1(t) + M_1u(t),~\text{ODE subsystem}, \\
\Delta_2:~~x_2(t) &= N_2x_1(t) + M_2u(t),~\text{AC subsystem 1}, \\
\Delta_3:~~x_3(t) &= N_3x_1(t) + M_3u(t) + L_3\dot{x}_2(t),~\text{AC subsystem 2,} \\
~~~~~~~~~x(t) &= x_1(t) + x_2(t) + x_3(t), \\
x_1(t) &= P_0P_1x(t), ~N_1 = P_0P_1E_2^{-1}A_2, ~M_1 = P_0P_1E_2^{-1}B, \\
x_2(t) &= P_0Q_1x(t), ~N_2 = P_0Q_1E_2^{-1}A_2, ~M_2 = P_0Q_1E_2^{-1}B, \\
x_3(t) &= Q_0x(t), ~N_3 = Q_0P_1E_2^{-1}A_2, ~M_3 = Q_0P_1E_2^{-1}B, ~L_3 = Q_0Q_1.
\end{split}
\end{equation*}

\end{lemma}
Proof is given in Appendix \ref{proof:index-2}.
\begin{lemma}[Index-3 DAE decoupling~\cite{marz1996canonical, banagaaya2016index}]\label{lm:index-3-decoupling}
An index-3 DAE system described by (\ref{eq:dae}) can be decoupled into a decoupled system using the matrix chain defined by Equation (\ref{eq:matrix_chain}) and the admissible projectors in Definition \ref{def:admissible_projectors} as follows:
\begin{equation*}
\begin{split}
&~~~\Delta_1:~~\dot{x}_1(t) = N_1x_1(t) + M_1u(t),~\text{ODE subsystem}, \\
&~~~\Delta_2:~~x_2(t) = N_2x_1(t) + M_2u(t),~\text{AC subsystem 1}, \\
&~~~\Delta_3:~~x_3(t) = N_3x_1(t) + M_3u(t) + L_3\dot{x}_2(t),~\text{AC subsystem 2} \\
&~~~\Delta_4:~~x_4(t) = N_4x_1(t) + M_4u(t) + L_4\dot{x}_3(t) + Z_4\dot{x}_2(t),~\text{AC subsystem 3} \\
&~~~~~~~~~~~x(t) = x_1(t) + x_2(t) + x_3(t) + x_4(t),~\text{where:} \\
&x_1(t) = P_0P_1P_2x(t),~N_1 = P_0P_1P_2E_3^{-1}A_3,~M_1 = P_0P_1P_2E_3^{-1}B, \\
&x_2(t) = P_0P_1Q_2x(t), ~N_2 = P_0P_1Q_2E_3^{-1}A_3, ~M_2 = P_0P_1Q_2E_3^{-1}B, \\
&x_3(t) = P_0Q_1x(t), ~N_3 = P_0Q_1P_2E_3^{-1}A_3, ~M_3 = P_0Q_1P_2E_3^{-1}B, ~L_3 = P_0Q_1Q_2, \\
&x_4(t) = Q_0x(t), ~N_4 = Q_0P_1P_2E_3^{-1}A_3, ~M_4 = Q_0P_1P_2E_3^{-1}B, ~L_4 = Q_0Q_1, ~Z_4 = Q_0P_1Q_2.
\end{split}
\end{equation*}
\end{lemma}
Proof is given in Appendix \ref{proof:index-3}.

It should be noted that the AC subsystems $\Delta_3$ and $\Delta_4$ in Lemma \ref{lm:index-2-decoupling} and \ref{lm:index-3-decoupling} are called algebraic constraints, though they contain the derivatives of $x_2(t)$ and $x_3(t)$. This is because the explicit forms of these algebraic constraints can be obtained if we further extend the derivatives using the corresponding ODE subsystems. In addition, one can see that for a DAE system with index-$2$ or -$3$, a set of admissible projectors need to be constructed for decoupling. In the following, we give a Proposition and Lemmas that are used to construct such admissible projectors.

\begin{proposition}[Orthogonal projector on a subspace]\label{pro:orthogonal-projector}
Given a real matrix $Z \in \mathbb{R}^{n \times n}$ such that $rank(Z) = r < n$, the Singular-Value Decomposition (SVD) of $Z$ has the form:
\begin{equation}
Z = [L_1~L_2]\begin{bmatrix}S_{r\times r} & 0 \\ 0 & 0 \\\end{bmatrix}\begin{bmatrix} K_1^{\textbf{T}}\\ K_2^{\textbf{T}} \\\end{bmatrix},
\end{equation}
where $L_1, K_1 \in \mathbb{R}^{n \times r}$ and $L_2, K_2 \in \mathbb{R}^{n \times n - r}$. Then, the matrix $Q = K_2K_2^{\textbf{T}}$ is an orthogonal projector on $Ker(Z)$, i.e., $ZQ = 0$, $Q = Q^{\textbf{T}}$ and $Q^2 = Q$.
\end{proposition}
Proof is given in Appendix \ref{proof:orthogonal-projector}.

For an index-$2$ or -$3$ DAE system, using Proposition \ref{pro:orthogonal-projector}, we can construct a set of projectors of the matrix chain defined in Equation (\ref{eq:matrix_chain}). However, these projectors are not yet admissible, because $Q_jQ_i \neq 0, j > i$. Instead, the admissible projectors can be constructed based on these inadmissible projectors using the following Lemmas.
\begin{lemma}[Admissible projectors for an index-2 DAE system]\label{lm:admissible-projectors-index-2}
Given an index-2 DAE system described by (\ref{eq:dae}), let $Q_0$ and $Q_1$ respectively be the orthogonal projectors of $E_0$ and $E_1$ of the matrix chain defined in Equation (\ref{eq:matrix_chain}). The following projectors $Q_0^*$ and $Q_1^*$ are admissible: $Q_0^* = Q_0,~Q_1^* = -Q_1E_2^{-1}A_1$.
\end{lemma}
Proof is given in Appendix \ref{proof:admissible-projectors-index-2}.
\begin{lemma}[Admissible projectors for an index-3 DAE system]\label{lm:admissible-projectors-index-3}
Given an index-3 DAE system described by (\ref{eq:dae}), let $Q_0$, $Q_1$ and $Q_2$ respectively be the orthogonal projectors of $E_0$, $E_1$ and $E_2$ of the matrix chain defined in Equation (\ref{eq:matrix_chain}). We define the following projectors and the corresponding new matrices for the matrix chain as:
\vspace{-0.5em}
\begin{equation*}
\begin{split}
Q_2^{\prime} = -Q_2E_3^{-1}A_2,~Q_1^{\prime} = -Q_1P_2^{\prime}E_3^{-1}A_1,~E_2^{\prime} = E_1 - A_1Q_1^{\prime},~A_2^{\prime} = A_1P_1^{\prime}
\end{split}
\end{equation*}
where $P_2^{\prime} = I_n - Q_2^{\prime}$ and $P_1^{\prime} = I_n - Q_1^{\prime}$.
Let $Q_2^{\prime\prime}$ be the orthogonal projector on $E_2^{\prime}$ and $E_3^{\prime\prime} = E_2^{\prime} - A_2^{\prime}Q_2^{\prime\prime}$, then the following projectors $Q_0^*, Q_1^*$ and $Q_2^*$ are admissible: $Q_0^* = Q_0,~Q_1^* = Q_1^{\prime},~Q_2^* = -Q_2^{\prime\prime}(E_3^{\prime\prime})^{-1}A_2^{\prime}$.
\end{lemma}
Proof is given in Appendix \ref{proof:admissible-projectors-index-3}.

Lemmas \ref{lm:admissible-projectors-index-2} and \ref{lm:admissible-projectors-index-3} are the constructions of admissible projectors for index-$2$ and -$3$ DAE systems. Algorithm \ref{alg:admissible-projectors} clarifies how to construct admissible projectors for a DAE system with an index up to $3$. Next, based on the decoupled DAE system, we discuss the consistent condition of the system and analyze the system behavior under the effect of input functions.

\begin{algorithm}[t]
    \label{alg:admissible-projectors}
    \caption{Admissible Projectors Construction}
    \textbf{Input}: $(E, A)$ \% matrices of a DAE system \\
    \textbf{Output}: admissible projectors

    \begin{algorithmic}[1]
        \Procedure{Initialization}{}
        \State projectors = [~] \% a list of projectors
        \State $E_0 = E,~A_0 = A$ and $n = number~of~state~variables$
        \EndProcedure

        \Procedure{Construction of admissible projectors}{}
        \State ~~~\textbf{if} \textit{rank($E_0$) == n}:
        \State ~~~~~~exit() \% $E$ is nonsingular, thus, the DAE is equivalent to an ODE.
        \State ~~~\textbf{else}:
        \State ~~~~~~$Q_0$ = \textit{orthogonal\_projector\_on\_Ker($E_0$)},~$P_0 = I_n - Q_0$,~$E_1 = E_0 - A_0Q_0$
        \State ~~~~~~\textbf{if} \textit{rank($E_1$) == n}:
        \State ~~~~~~~~~projectors $\leftarrow Q_0$ \% the DAE has index-1
        \State ~~~~~~\textbf{else}:
        \State ~~~~~~~~~$Q_1$ = \textit{orthogonal\_projector\_on\_Ker($E_1$)},~$P_1 = I_n - Q_1$
        \State ~~~~~~~~~$A_1 = A_0P_0,~E_2 = E_1 - A_1Q_1$
        \State ~~~~~~~~~\textbf{if} \textit{rank($E_2$) == n}:
        \State ~~~~~~~~~~~~$Q_1^* = -Q_1E_2^{-1}A_1$
        \State ~~~~~~~~~~~~projectors $\leftarrow (Q_0,~Q_1^*)$ \% the DAE has index-2
        \State ~~~~~~~~~\textbf{else}:
        \State ~~~~~~~~~~~~$Q_2$ = \textit{orthogonal\_projector\_on\_Ker($E_2$)},~$P_2 = I_n - Q_2$
        \State ~~~~~~~~~~~~$A_2 = A_1P_1,~E_3 = E_2 - A_2Q_2$
        \State ~~~~~~~~~~~~~~~\textbf{if} \textit{rank($E_3$) == n}:
        \State ~~~~~~~~~~~~~~~~~~$Q_2^{\prime} = Q_2E_3^{-1}A_2$, ~$P_2^{\prime} = I_n - Q_2^{\prime}$, ~$Q_1^{\prime} = Q_1P_2^{\prime}E_3^{-1}A_1$
        \State ~~~~~~~~~~~~~~~~~~$E_2^{\prime} = E_1 - A_1Q_1^{\prime}$, $P_1^{\prime} = I_n - Q_1^{\prime}$, ~$A_2^{\prime} = A_1P_1^{\prime}$
        \State ~~~~~~~~~~~~~~~~~~$Q_2^{\prime\prime}$ = \textit{orthogonal\_projector\_on\_Ker($E_2^{\prime}$)},~$P_2^{\prime\prime} = I_n - Q_2^{\prime\prime}$
        \State ~~~~~~~~~~~~~~~~~~$E_3^{\prime\prime} = E_2^{\prime} - A_2^{\prime}Q_2^{\prime\prime}$,~$Q_2^* = -Q_2^{\prime\prime}(E_3^{\prime\prime})^{-1}A_2^{\prime}$
        \State ~~~~~~~~~~~~~~~~~~projectors $\leftarrow (Q_0,~Q_1^{\prime}, Q_2^*)$ \% the DAE has index-3
        \State ~~~~~~~~~~~~~~~\textbf{else}:
        \State ~~~~~~~~~~~~~~~~~~exit() \% the DAE has index lager than 3
        \State ~~~\textbf{return} projectors
        \EndProcedure
    \end{algorithmic}
\end{algorithm}

\vspace{-1em}
\section{Consistency}
\label{sec:consistency}
\vspace{-0.5em}
In this section, we discuss the consistent condition for a DAE system. Using the decoupled DAE system, the consistent condition for the initial state and inputs is derived. Additionally, the piecewise constant assumption on the inputs used in \cite{bak2017simulation} for ODE systems may lead to \emph{impulsive behavior} in high-index DAE systems. To avoid this, we limit our problem to \emph{smooth} and \emph{specific-user-defined} inputs. As a result, DAE systems with inputs can be converted to autonomous DAE systems, where \emph{consistent spaces} for the initial states and inputs can be conveniently defined and checked. Furthermore, the reachable set computation is executed efficiently using a decoupled autonomous DAE system.

Using Lemmas \ref{lm:index-1-decoupling}, \ref{lm:index-2-decoupling}, and \ref{lm:index-3-decoupling}, to guarantee a solution for the DAE system, the initial states and inputs must satisfy the following conditions:
\vspace{-0.4em}
\begin{equation}\label{eq:consistent}
\begin{split}
\text{Index-1 DAE}:~&x_2(0) = N_2x_1(0) + M_2u(0), \\
\text{Index-2 DAE}:~&x_2(0) = N_2x_1(0) + M_2u(0), \\
~~~~~~~~~~~~~~~~~~~~&x_3(0) = N_3x_1(0) + M_3u(0) + L_3\dot{x}_2(0), \\
\text{Index-3 DAE}:~&x_2(0) = N_2x_1(0) + M_2u(0), \\
~~~~~~~~~~~~~~~~~~~~&x_3(0) = N_3x_1(0) + M_3u(0) + L_3\dot{x}_2(0), \\
~~~~~~~~~~~~~~~~~~~~&x_4(0) = N_4x_1(0) + M_4u(0) + L_4\dot{x}_3(0) + Z_4\dot{x}_2(0).
\end{split}
\end{equation}

Assuming that the consistent condition is satisfied, Lemmas \ref{lm:index-2-decoupling} and \ref{lm:index-3-decoupling} indicate the solution of the system involves the derivatives of the input functions $\dot{x}_2(t) = N_2\dot{x}_1(t) + M_2\dot{u}(t)$ and $\dot{x}_3(t) = N_3\dot{x}_1(t) + M_4\dot{u}(t) + L_3[N_2\ddot{x}_1(t) + M_2\ddot{u}(t)]$. In cases where we apply piecewise constant inputs to a high-index DAE system, the impulsive behavior may appear in the system at an exact discrete time point $t_k$. For example, let $u(t)$ be a step function in $[t_k, t_{k + 1})$, then $\dot{u}(t_k) = \delta(t_k)$, where $\delta(t_k)$ is the Dirac function describing an impulse. To avoid such impulsive behavior and do reachability analysis for high-index DAE systems, we limit our approach to smooth inputs which are governed by the following ODE: $\dot{u}(t) = A_uu(t),~u(0) = u_0 \in U_0$, where $A_u \in \mathbb{R}^{m \times m}$ is the user-defined input matrix, and $U_0$ is the set of initial inputs.
\begin{remark}
By introducing the input matrix $A_u$, we limit the safety verification and falsification of a high-index DAE system to a class of \emph{specific-user-defined} inputs. If $A_u = 0$, then the input set is a set of constant inputs.
\end{remark}

Given a user-defined input matrix $A_u$, a DAE system described by (\ref{eq:dae}) can be converted to an equivalent autonomous DAE system of the following form:
\begin{equation}\label{eq:auto-dae}
\bar{E}\dot{\bar{x}}(t) = \bar{A}\bar{x}(t),
\end{equation}
where $\bar{x}(t) = \begin{bmatrix} x(t) \\ u(t) \\\end{bmatrix} \in \mathbb{R}^{n + m} $, $\bar{E} = \begin{bmatrix} E & 0 \\ 0 & I_{m} \end{bmatrix},~\bar{A} = \begin{bmatrix} A & B \\ 0 & A_u \end{bmatrix} \in \mathbb{R}^{(n + m) \times (n + m)}$ and the state of the original DAE is: $x(t) = [I_n~~0]\bar{x}(t)$.

Similar to the original DAE system, the autonomous DAE system (\ref{eq:auto-dae}) can be decoupled to form one autonomous ODE subsystem and one or several AC subsystems. It should be noted that the autonomous DAE system has the same index as the original one.

We have discussed the conversion of a DAE system with user-defined input to an autonomous DAE system. Next, we derive the \emph{consistent space} for the initial condition of an autonomous DAE system. All previous results apply to these systems given that $u(t) = 0$.
\begin{definition}[Consistent Space for an autonomous DAE system]\label{def:consistent-space}
Consider an autonomous DAE system ($\Delta$) defined in Equation (\ref{eq:dae}) by letting $u(t) = 0$. From this, we define in the following a ``consistent matrix'' $\Gamma$ as:
\vspace{-0.5em}
\begin{equation*}
\begin{split}
&\text{Index-1 $\Delta$}:~\Gamma = Q_0 - N_2P_0,~(Q_0, P_0, N_2) \text{are defined in Lemma \ref{lm:index-1-decoupling}}, \\
&\text{Index-2 $\Delta$}:~\Gamma = \begin{bmatrix} P_0Q_1 - N_2P_0P_1 \\ Q_0 - (N_3 + L_3N_2N_1)P_0P_1 \\\end{bmatrix}, \\
&~~~~~~~~~~~~~~~~(Q_i,P_i, N_i, L_i) \text{are defined in Lemma \ref{lm:index-2-decoupling}}, \\
&\text{Index-3 $\Delta$}:~\Gamma = \begin{bmatrix} P_0P_1Q_2 - N_2P_0P_1P_2 \\ P_0Q_1 - (N_3 + L_3N_2N_1)P_0P_1P_2 \\ Q_0 - [N_4 + L_4(N_3N_1 + L_3N_2N_1^2) + Z_4N_2N_1]P_0P_1P_2 \\ \end{bmatrix},\\
&~~~~~~~~~~~~~~~~(Q_i, P_i, N_i, L_i, Z_4) \text{are defined in Lemma \ref{lm:index-3-decoupling}},
\end{split}
\end{equation*}
then, $Ker(\Gamma)$ is the consistent space of the system $\Delta$, where $Ker(\Gamma)$ denotes the null space of the matrix $\Gamma$.
\end{definition}

An initial state $x_0$ is consistent if it is in the consistent space, i.e., $\Gamma x_0 = 0$. The consistent matrix and consistent space is introduced because it is useful and convenient for checking the consistency of an initial set of states represented using a star set. For example, assume that the initial set of states is defined by $\Theta(0) = \langle V(0), P\rangle$, then this set is consistent for all $\alpha$ satisfying the predicate $P$ if $\Gamma V(0) = 0$. This means that we require consistency for all points in the initial set. With a consistent initial set of states, we investigate the reachable set computation and safety verification/falsification of an autonomous DAE system in the next section.
\section{Reachability analysis}
\vspace{-0.2em}
\seclabel{reachability}
\subsection{Reachable set computation}
\vspace{-0.1em}
The reachable set of an autonomous DAE system is constructed by combining the reachable set of all of its decoupled subsystems. The reachable set of all AC subsystems can be derived from the reachable set of the ODE subsystem, which can be computed efficiently using existing ODE solvers. We first discuss the reachable set computation of the ODE subsystem by exploiting its \emph{superposition property}. Then, the reachable set of the autonomous DAE system is constructed conveniently using only matrix addition and multiplication.

Let $\Theta(0) = \langle V(0), P \rangle$ be the initial set of states of an autonomous DAE system defined in (\ref{eq:dae}) by letting $u(t) = 0$. Assume that the initial set of states, $X(0)$, satisfies the consistent condition. After decoupling, the initial set of states of the ODE subsystem $\Theta_1(0)$ is obtained as follows: $\Theta_1(0) = \langle V_1(0), P \rangle$ where $V_1(0) = (\prod_{i=0}^{\mu - 1}P_0\cdots P_{\mu - 1})V(0) = [v_1^1(0)~v_2^1(0)~\cdots~v_k^1(0)]$, $\mu$ is the index of the DAE system, and $P_i, (i = 0, \cdots, \mu - 1),$ are defined in Lemma \ref{lm:index-1-decoupling} or \ref{lm:index-2-decoupling} or \ref{lm:index-3-decoupling} corresponding to the index $\mu$.

Then, for any $x_1(0) \in \Theta_1(0)$, we have $x_1(0) = \Sigma_{i=1}^{k}\alpha_iv_i^1(0)$. The solution of the ODE subsystem at time $t$ is given by: $x_1(t) = \Sigma_{i=1}^{k}\alpha_iv_i^1(t) = V_1(t)\alpha$, where $v_i^1(t) = e^{N_1t}v_i^1(0)$ and $V_1(t) = [v_1^1(t)~v_2^1(t)~\cdots~v_k^1(t)]$. Therefore, the reachable set of the ODE subsystem at anytime $t$ is also a star set defined by $\Theta_1(t) = V_1(t)\alpha$. Using existing ode solvers, we can construct the matrix $V_1(t)$ at anytime $t$. From $\Theta_1(t)$, the reachable set of the autonomous DAE system can be obtained using the following Lemma.
\begin{lemma}[Reachable set construction]\label{lm:reachable-set}
Given an autonomous DAE system defined in Equation (\ref{eq:dae}) where $u(t) = 0$ and a consistent initial set of states $\Theta(0) = \langle V(0), P \rangle$, let $\Theta_1(t) = \langle V_1(t), P \rangle$ be the reachable set at time $t$ of the corresponding ODE subsystem after decoupling. Then, the reachable set $\Theta(t)$ at time $t$ of the system is given by $\Theta(t) = \langle V(t) = \Psi V_1(t), P \rangle $, where $\Psi$ is a ``reachable set projector'' defined below.
\begin{equation}
\begin{split}
&\text{Index-1}:~\Psi = (I_n + N_2),~\text{$N_2$ is defined in Lemma \ref{lm:index-1-decoupling}}, \\
&\text{Index-2}:~\Psi = (I_n + N_2 + N_3 + L_3N_2N_1), \\
&~~~~~~~~~~~~~~~\text{($N_{i=1,2,3}, L_3$) are defined in Lemma \ref{lm:index-2-decoupling}}, \\
&\text{Index-3}:~\Psi = (I_n + N_2 + N_3 + N_4 + L_3N_2N_1 + \\
&~~~~~~~~~~~~~~~~~~~~~~~~~~~L_4N_3N_1 + L_4L_3N_2N_1^2 + Z_4N_2N_1), \\
&~~~~~~~~~~~~~~~\text{$(N_{i=1,2,3,4}, L_{i=3,4}, Z_4)$ are defined in Lemma \ref{lm:index-3-decoupling}}.
\end{split}
\end{equation}
\end{lemma}
Proof is given in Appendix \ref{proof:reachable-set}.

The Algorithm \ref{alg:reachable-set-construction} describes how to construct a reachable set of an autonomous DAE system. Next, from the constructed reachable set, we discuss how to verify or falsify the system safety property.
\begin{algorithm}\label{alg:reachable-set-construction}
    \caption{Reachable set computation}
    \textbf{Inputs}: Matrices of an autonomous DAE system $(E, A)$, initial set of states $\Theta(0) = \langle V(0), P \rangle$, time step $h$, number of steps $N$.

    \textbf{Output}: Reachable set \% A list of stars

    \begin{algorithmic}[1]
        \Procedure{Initialization}{}
        \State $ListOfStars$ = [~]
        \State Decoupling the system (Section \ref{sec:decoupling})
        \State Obtain consistent space $Ker(\Gamma)$ (Definition \ref{def:consistent-space})
        \State \textbf{If} $V(0) \not \in Ker(\Gamma)$: exit() \% inconsistent initial set of states
        \State \textbf{Else}: Obtain initial set of states for ODE subsystem: \\
        ~~~~~~~~~~~~~~$\Theta_1(0) = \langle V_1(0), P \rangle$,~$V_1(0) = [v_1^1(0)~\cdots~v_k^1(0)]$
        \EndProcedure

        \Procedure{Reachable set construction}{}
        \State ~~~\textbf{for $j = 0, 1, 2, \cdots, N$}:
        \State ~~~~~~~~~\textbf{for $i = 1, 2, \cdots, k$}:
        \State ~~~~~~~~~~~~~~~Compute $v_i^1(jh) = e^{N_1jh}v_i^1(0)$    \% using ODE solvers
        \State ~~~~~~~~~Construct $V_1(jh) = [v_1^1(jh)~v_2^1(jh)~\cdots~v_k^1(jh)]$
        \State ~~~~~~~~~Compute $V(jh)$ from $V_1(jh)$ using Lemma \ref{lm:reachable-set}
        \State ~~~~~~~~~Construct $\Theta(jh) = \langle V(jh), P \rangle$
        \State ~~~~~~~~~$ListOfStars \leftarrow \Theta(jh)$
        \State ~~~\textbf{return} $ListOfStars$
        \EndProcedure
    \end{algorithmic}
\end{algorithm}

\subsection{Safety verification and falsification}

By utilizing the star set to represent the reachable set of a DAE system, the safety verification and falsification problem is solved in the following manner. Let $Unsafe(\Delta) \triangleq Gx \leq f$ be the unsafe set of an autonomous DAE system and assume that we want to check the safety of the system at the time step $t_j = jh$. This is equivalent to checking $GV(jh)\alpha \leq f$ subject to $P(\alpha) \triangleq C\alpha \leq d$, where $V(jh)$ is the basic matrix of the reachable set $\Theta(jh)$ of the system at time $jh$ computed using Algorithm \ref{alg:reachable-set-construction}. Combining these constraints, the problem changes to checking the feasibility of the following linear predicate: $\bar{P} \triangleq\bar{G}\alpha \leq \bar{f}$, where $\bar{G} = [(GV(jh))^T~~C^T]^T$ and $\bar{f} = [f^T~~d^T]^T$. This can be solved efficiently using existing linear programming toolboxes. Algorithm \ref{alg:verification-falsification} illustrates how to verify or falsify the safety property of an autonomous DAE system. In the next section, we evaluate our approach using a set of DAE benchmarks with several to thousands of states.

\begin{algorithm}[t]\label{alg:verification-falsification}
    \caption{Bounded-time safety verification/falsification}
    \textbf{Inputs}: $Reachable\_Set$ \% a list of stars; $Unsafe(\Delta) \triangleq Gx \leq f$ \% the unsafe set

    \textbf{Output}: $Safe/Unsafe$ and $Unsafe\_Trace$

    \begin{algorithmic}[1]
        \Procedure{Initialization}{}
        \State $N$ = number of stars in the reachable set
        \State $Status = Safe$
        \State $Unsafe\_Trace = [~]$
        \EndProcedure
        \Procedure{Verification/Falsification}{}
        \State ~~~\textbf{for $j = 1, 2, \cdots, N$}:
        \State ~~~~~~~~~$\Theta_j = Reachable\_Set[j] = \langle V_j, P \rangle$, $P \triangleq C\alpha \leq d$
        \State ~~~~~~~~~~Construct $\bar{P} \triangleq \begin{bmatrix} GV_j \\ C \\\end{bmatrix}\alpha \leq \begin{bmatrix} f \\ d \\\end{bmatrix}$
        \State ~~~~~~~~~\textbf{If} $\bar{P}$ is feasible:
        \State ~~~~~~~~~~~~~~~$Status = Unsafe$,~get $\alpha_{feasible}$, exit()
        \State ~~~\textbf{If} $Status = Unsafe$:
        \State ~~~~~~~~~\textbf{for $j = 1, 2, \cdots, N$}:
        \State ~~~~~~~~~~~~~~~Compute $x_j = V_j\alpha_{feasible}$
        \State ~~~~~~~~~~~~~~~$Unsafe\_Trace \leftarrow x_j$
        \State ~~~\textbf{return} $Status$,~$Unsafe\_Trace$
        \EndProcedure
    \end{algorithmic}
\end{algorithm}

\section{Experimental Results}
\seclabel{evaluation}
The detailed description of how our approach works  using the index-2, interconnected rotating masses system \cite{schon2003modeling} is presented in Appendix \ref{description:ex1}. In this section, we first analyze the time performance of our approach using the index-2, two-dimensional semidiscretized Stokes Equation benchmark \cite{mehrmann2005balanced}. Then, we test and compare the scalability of our approach with the well-known tool SpaceEx \cite{frehse2011spaceex} using a set of benchmarks. Using SpaceEx for DAE systems verification is non-trivial. To do this, we first decouple a DAE and check the consistency condition of the initial set of states and inputs. Then, we construct an automaton with the decoupled ODE subsystem. Finally, the state of the DAE system is declared as a set of invariants of the automaton using the ``reachable set projector'' $\Psi$ in Lemma \ref{lm:reachable-set}. Our approach is implemented in a tool called \emph{Daev}\footnote{https://www.dropbox.com/s/w1mrbl3nw31n2gs/daev.zip?dl=0} using python and its standard packages numpy, scipy, and mathplotlib. All experiments were done on a computer with the following configuration: Intel Core i7-6700 CPU @ 3.4GHz × 8 Processor, 62.8 GiB Memory, 64-bit Ubuntu 16.04.3 LTS OS.

\begin{example}[Semidiscretized Stokes Equation \cite{mehrmann2005balanced}]
This example studies the safety of a Stokes equation that describes the flow of an incompressible fluid in a two-dimensional spatial domain $\Omega$. The mathematical description of the Stokes-equation is given in Appendix \ref{description:ex2}. An index-2 DAE system is derived from the Stokes-equation by discretizing the domain $\Omega$ by a number of uniform square cells. Let $n$ be the number of discretized segments of the domain on the x- or y-axes, then the dimension of the DAE system is $3n^2 + 2n$. Additionally, we are interested in the velocity along the x- and y- axes, $v_x^c(t)$ and $v_y^c(t)$, of the fluid in the \emph{central cell} of the domain $\Omega$. The unsafe set of the system is defined: $Unsafe \triangleq -v_x^c(t) - v_y^c(t) \leq 0.04$.

By increasing the number of cells used to discretize the domain $\Omega$, we can produce an index-2 DAE system with arbitrarily large dimension. We evaluate the time performance of our approach via three scenarios. First, we discuss how the times for decoupling, reachable set computation, and safety checking are affected by changes in the system dimension. Second, we analyze the reachable set computation time along with the \emph{width} of the basic matrix of the initial set $V(0)$, i.e., the number of the initial basic vectors. Finally, because the reachable set of the system is constructed from the reachable set of its corresponding ODE subsystem, which is computed using ODE solvers as shown in Algorithm \ref{alg:reachable-set-construction}, we investigate the time performance of reachable set computation using different ODE solving schemes.

Table \ref{tab:verification-time-versus-dimensions} presents the verification time, \emph{V-T}, for the Stokes-equation benchmark with different dimensions. The verification time is broken into three components measured in seconds: decoupling time \emph{D-T}, reachable set computation time \emph{RSC-T}, and checking safety time \emph{CS-T}. Table \ref{tab:verification-time-versus-dimensions} shows the decoupling and reachable set computation times dominate the time for verification process. In addition, these times increase as the system size grows. The time for checking safety is almost unchanged and  very small. This happens because the size of the feasibility problem $\bar{P}$ defined in Algorithm \ref{alg:verification-falsification} is unchanged and usually small when we only check the safety in some specific directions defined by the unsafe matrix $G$ in the Algorithm.
\begin{table}[t]
\caption{\bf Verification time of Stokes-equation with different dimensions $n$.}
\label{tab:verification-time-versus-dimensions}
\centering 
\begin{tabular}{l@{\hskip 0.15in}l@{\hskip 0.15in}l@{\hskip 0.15in}l@{\hskip 0.15in}l@{\hskip 0.15in}l@{\hskip 0.15in}l}
\toprule
{\bf n}& {\bf $86$} & {\bf $321$} & {\bf $706$} & {\bf $1241$} & {\bf $1926$} & {\bf $2761$} \\
\toprule
{\bf D-T}& $0.012s$ & $0.63s$ & $6.32s$ & $40.38s$ &  $155.32s$ & $466.38s$  \\
\midrule
{\bf RSC-T}& $0.019s$ & $0.37s$ & $2.98s$ & $19.29s$ &  $68.15s$ & $200.89s$  \\
\midrule
{\bf CS-T}& $0.0017s$ & $0.0014s$ & $0.0015s$ & $0.0017s$ &  $0.0018s$ & $0.002s$  \\
\midrule
{\bf V-T}& $0.0327s$ & $1.0014s$ & $9.3015s$ & $59.6717s$ &  $223.4718s$ & $667.272s$  \\
\bottomrule
\end{tabular}
\end{table}

Since the reachable set the Stokes-equation benchmark is constructed by simulating its corresponding ODE subsystem with each initial vector of its initial basic matrix, the time for computing the reachable set of the Stokes-equation depends linearly on the number of the initial basic vectors $k$. Table \ref{tab:RSC-T-versus-k} shows the reachable set computation time, $RSC-T$, for the Stokes-equation of dimension $n = 321$ versus the number of the initial basic vectors $k$.

\begin{table}[t]
\caption{\bf Reachable set computation time of Stokes-equation of dimensions $n = 321$ with different number of initial basic vectors $k$.}
\label{tab:RSC-T-versus-k}
\centering 
\begin{tabular}{l@{\hskip 0.15in}l@{\hskip 0.15in}l@{\hskip 0.15in}l@{\hskip 0.15in}l@{\hskip 0.15in}l@{\hskip 0.15in}l@{\hskip 0.15in}l}
\toprule
{\bf k}& {\bf $2$} & {\bf $4$} & {\bf $6$} & {\bf $8$} & {\bf $10$} & {\bf $12$} & {\bf $14$}\\
\toprule
{\bf RSC-T}& $1.9s$ & $3.41s$ & $5.01s$ & $6.71s$ &  $8.3s$ & $9.9s$ &  $11.44s$ \\
\bottomrule
\end{tabular}
\end{table}
Our approach relies on existing ODE solvers. Therefore, it is interesting to consider how the reachable set computation time performs with different existing ODE solving schemes supported by the \emph{scipy} package such as \emph{vode}, \emph{zvode}, \emph{lsoda}, \emph{dopri5} and \emph{dop853}. All solvers are used with the absolute tolerance $atol = 1e$-$12$ and the relative tolerance $rtol = 1e$-$08$. Figure \ref{fig:ex2-reach-time-vs-solvers} illustrates the time performance of different schemes and indicates that the \emph{vode}, \emph{dopri5}, and \emph{dop853} are fast schemes that should be used for large DAE systems. In addition, we should avoid using the \emph{lsoda} and \emph{zvode} schemes for large DAE systems due to theirs slow performance.
\begin{figure}[t]
    \centering
        \includegraphics[height=3in]{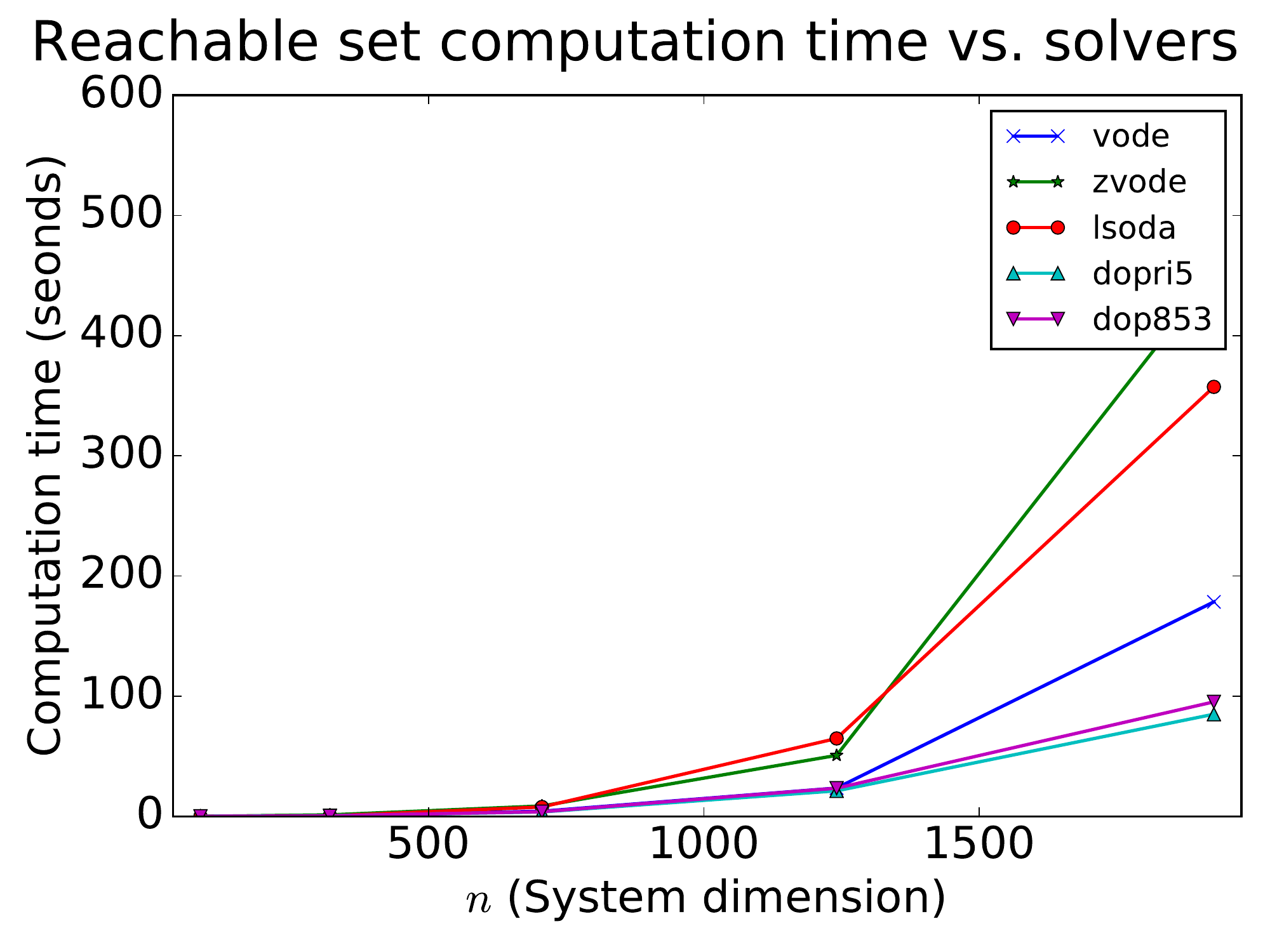}
        \vspace{-0.5em}
        \caption{Reachable set computation time of Stokes-equation using different ode solvers}
        \label{fig:ex2-reach-time-vs-solvers}
\end{figure}
\end{example}
\begin{example}[Scalability comparison]
Table \ref{tab:verification-results} presents the verification results for all benchmarks using Daev and SpaceEx. For large benchmarks such as PEEC, MNA-1, MNA-4, and Stokes-equation, SpaceEx can not parse the large model files so their results are not included in the table. We can see from the table that the proposed approach is faster than the over-approximation approach for most of DAE systems with small and medium dimensions. This result is trivial because we compute the reachable set of DAE systems in \emph{discrete time} using ODE solvers while another one computes the reachable set in \emph{continuous time}. The main benefit of our approach is its scalability for large DAE systems with thousands of state variables where the over-approximation approach is not applicable. Furthermore, it can produce an unsafe trace in the case that a DAE system violates its safety property as demonstrated in Example \ref{ex1}. Therefore, our approach is practically useful for falsification of large, linear DAE systems.

It is interesting that Daev and SpaceEx produce different results for the case of the interconnected rotating mass benchmark. The Daev tool gives an unsafe result when checking $x_3 \leq -0.9$ and produces an unsafe trace, depicted in Appendix \ref{description:ex1}, using the \emph{dopri5} scheme. However, SpaceEx gives a safe result in this case. Additionally, when checking $x_4 \leq -1.0$, SpaceEx produces an unsafe result while Daev obtains a safe result. While the latter case is understandable since SpaceEx computes a conservative, over-approximated reachable set that may contain unreachable states, the mismatch in the first case between two tools is interesting and hard to explain.

\begin{table}[p]
\caption{\bf Verification results for all benchmarks.}
\label{tab:verification-results}
\centering 
\begin{tabular}{p{3.35cm}@{\hskip 0.15in}l@{\hskip 0.1in}l@{\hskip 0.15in}l@{\hskip 0.15in}l@{\hskip 0.15in}l@{\hskip 0.15in}l@{\hskip 0.15in}l}
\toprule
{\bf Benchmarks}& {\bf n} & {\bf Index} & {\bf Unsafe Set} & {\bf Tool} & {\bf Result} & {\bf V-T(s)}\\
\toprule
{\bf RL network} \cite{ho1975modified} & $3$ & $2$ & $x_1 \leq -0.2 \wedge x_2 \leq -0.1$ & Daev  &   unsafe &  $0.184$ \\
                &     &     &                & SpaceEx   & unsafe &  $2.002$ \\

{\bf RL network$^*$} \cite{ho1975modified}& $3$ & $2$ & $x_1 \geq 0.2$ & Daev    & safe &  $0.44$ \\
                     & &  &                    & SpaceEx    & safe &  $0.502$ \\

\midrule
{\bf RLC circuit} \cite{dai1989singular} & $4$ & $1$ & $x_1 + x_3 \geq  0.2$ & Daev  & unsafe &  $0.224$ \\
                &     &     &                & SpaceEx  & unsafe &  $2.902$ \\
{\bf RLC circuit$^*$} \cite{dai1989singular} & $4$ & $1$ & $x_4 \leq -0.3$ & Daev & safe &  $1.37$ \\
                &     &     &                & SpaceEx    & safe &  $0.602$ \\

\midrule
{\bf Interconnected rotating mass} \cite{schon2003modeling}& $4$ & $2$ & $x_3 \leq -0.9$ & Daev  & unsafe &  $0.37$ \\
                &     &     &                                    & SpaceEx  & safe &  $0.802$ \\

{\bf Interconnected rotating mass$^*$} \cite{schon2003modeling}& $4$ & $2$ & $x_4 \leq -1.0$ &Daev  & safe &  $0.114$ \\
                &     &     &                                        & SpaceEx   & unsafe &  $1.02$ \\

\midrule
{\bf Generator} \cite{gerdin2004parameter}& $9$ & $3$ & $x_9 \geq 0.01$ & Daev  &  unsafe &  $0.4$ \\
                &     &     &                & SpaceEx  & unsafe &  $10.02$ \\

{\bf Generator$^*$} \cite{gerdin2004parameter}& $9$ & $3$ & $x_1 \geq 1.0$ & Daev  & safe &  $0.684$ \\
                &     &     &                & SpaceEx   & safe &  $1.602$ \\

\midrule
{\bf Damped-mass spring} \cite{mehrmann2005balanced} & $11$ & $3$ & $x_3 \leq 1 \wedge x_8 \leq 1.5$ & Daev & safe &  $1.06$ \\
                &     &     &                        & SpaceEx  & safe &  $2.31$ \\
{\bf Damped-mass spring$^*$} \cite{mehrmann2005balanced}& $11$ & $3$ & $x_8 \leq -0.2$ & Daev     &  unsafe &  $1.08$ \\
                &     &     &                        & SpaceEx  & unsafe &  $2.81$ \\

\midrule
{\bf PEEC} \cite{chahlaoui2002collection}& $480$ & $2$ & $x_{478} \geq 0.05$ & Daev & safe &  $28.84$ \\
{\bf PEEC$^*$} \cite{chahlaoui2002collection}& $480$ & $2$ & $x_{478} \geq 0.01$ & Daev & unsafe &  $28.25$ \\
\midrule
{\bf MNA-1} \cite{chahlaoui2002collection}& $578$ & $2$ & $x_1 \geq -0.001$ & Daev &  safe &  $192.7$ \\
{\bf MNA-1} \cite{chahlaoui2002collection}& $578$ & $2$ & $x_1 \geq -0.0015$ & Daev &  unsafe &  $202.6$ \\
\midrule
{\bf MNA-4} \cite{chahlaoui2002collection}& $980$ & $3$ & $x_2 \geq 0.0005$ & Daev  &  safe &  $1858.4$ \\
{\bf MNA-4} \cite{chahlaoui2002collection}& $980$ & $3$ & $x_2 \geq 0.0002$ & Daev  &  unsafe &  $1836.04$ \\
\midrule
{\bf Stokes-equation} \cite{mehrmann2005balanced}& $4880$ & $2$ & $v_x^c + v_y^c \leq -0.04$ &  Daev  & unsafe &  $3502.3$ \\
{\bf Stokes-equation$^*$} \cite{mehrmann2005balanced}& $4880$ & $2$ & $v_x^c \geq 0.2$ &  Daev  & safe &  $3532.3$ \\
\bottomrule
\end{tabular}
\end{table}

\end{example}

\section{Conclusion and Future Work}
\seclabel{conclusion}
\vspace{-0.25em}
We have studied a simulation-based reachability analysis for high-index, linear DAE systems. The experiential results show that our approach can deal with DAE systems with up to thousands of state variables. Therefore, it is useful and applicable to verify or falsify safety-critical CPS involving DAE dynamics. Additionally, the decoupling and the consistency checking techniques used in our approach can be used as a transformation pass for existing over-approximation techniques \cite{althoff2015introduction, frehse2011spaceex} to verify the safety of DAE systems with small and medium dimension.

The reachability analysis for DAE systems with millions of dimensions remains challenging. The verification time of our approach depends mostly on the decoupling and the reachable set computation times. Therefore, to enhance the time performance and the scalability of our approach to make it work for million-dimensional DAE systems, both decoupling and reachable set computation techniques need to be improved. A promising application that inspires seeking a such scalable approach is verification and falsification of very large circuits which are described in the form of high-index DAE.


\normalsize
\let\oldbibliography\thebibliography
\renewcommand{\thebibliography}[1]{\oldbibliography{#1}
\setlength{\itemsep}{0pt}} 
\bibliographystyle{splncs03}
\bibliography{tran}  

\begin{thebibliography}{10}
\providecommand{\url}[1]{\texttt{#1}}
\providecommand{\urlprefix}{URL }

\bibitem{althoff2014tac}
Althoff, M., Krogh, B.: Reachability analysis of nonlinear
  differential-algebraic systems. Automatic Control, IEEE Transactions on
  59(2),  371--383 (2014)

\bibitem{althoff2015introduction}
Althoff, M.: An introduction to cora 2015. In: Proc. of the Workshop on Applied
  Verification for Continuous and Hybrid Systems (2015)

\bibitem{bak2017simulation}
Bak, S., Duggirala, P.S.: Simulation-equivalent reachability of large linear
  systems with inputs. In: International Conference on Computer Aided
  Verification. pp. 401--420. Springer (2017)

\bibitem{banagaaya2016index}
Banagaaya, N., Al{\`\i}, G., Schilders, W.H.: Index-aware model order reduction
  methods. Springer (2016)

\bibitem{byrne1988differential}
Byrne, G., Ponzi, P.: Differential-algebraic systems, their applications and
  solutions. Computers \& chemical engineering  12(5),  377--382 (1988)

\bibitem{chahlaoui2002collection}
Chahlaoui, Y., Van~Dooren, P.: A collection of benchmark examples for model
  reduction of linear time invariant dynamical systems.  (2002)

\bibitem{chen2013flow}
Chen, X., {\'A}brah{\'a}m, E., Sankaranarayanan, S.: Flow*: An analyzer for
  non-linear hybrid systems. In: International Conference on Computer Aided
  Verification. pp. 258--263. Springer (2013)

\bibitem{dai1989singular}
Dai, L.: Singular control systems (lecture notes in control and information
  sciences)  (1989)

\bibitem{duggirala2015c2e2}
Duggirala, P.S., Mitra, S., Viswanathan, M., Potok, M.: C2e2: a verification
  tool for stateflow models. In: International Conference on Tools and
  Algorithms for the Construction and Analysis of Systems. pp. 68--82. Springer
  (2015)

\bibitem{duggirala2016parsimonious}
Duggirala, P.S., Viswanathan, M.: Parsimonious, simulation based verification
  of linear systems. In: International Conference on Computer Aided
  Verification. pp. 477--494. Springer (2016)

\bibitem{fan2016automatic}
Fan, C., Qi, B., Mitra, S., Viswanathan, M., Duggirala, P.S.: Automatic
  reachability analysis for nonlinear hybrid models with c2e2. In:
  International Conference on Computer Aided Verification. pp. 531--538.
  Springer (2016)

\bibitem{frehse2011spaceex}
Frehse, G., Le~Guernic, C., Donz{\'e}, A., Cotton, S., Ray, R., Lebeltel, O.,
  Ripado, R., Girard, A., Dang, T., Maler, O.: Spaceex: Scalable verification
  of hybrid systems. In: Computer Aided Verification. pp. 379--395. Springer
  (2011)

\bibitem{gerdin2004parameter}
Gerdin, M.: Parameter estimation in linear descriptor systems. Citeseer (2004)

\bibitem{girard2005reachability}
Girard, A.: Reachability of uncertain linear systems using zonotopes. In:
  Hybrid Systems: Computation and Control, pp. 291--305. Springer (2005)

\bibitem{leguernic2010}
Guernic, C.L., Girard, A.: Reachability analysis of linear systems using
  support functions. Nonlinear Analysis: Hybrid Systems  4(2),  250--262 (2010)

\bibitem{han2006reachability}
Han, Z., Krogh, B.H.: Reachability analysis of large-scale affine systems using
  low-dimensional polytopes. In: Hybrid Systems: Computation and Control, pp.
  287--301. Springer (2006)

\bibitem{ho1975modified}
Ho, C.W., Ruehli, A., Brennan, P.: The modified nodal approach to network
  analysis. IEEE Transactions on circuits and systems  22(6),  504--509 (1975)

\bibitem{kong2015dreach}
Kong, S., Gao, S., Chen, W., Clarke, E.: dreach: $\delta$-reachability analysis
  for hybrid systems pp. 200--205 (2015)

\bibitem{marz1996canonical}
M{\"a}rz, R.: Canonical projectors for linear differential algebraic equations.
  Computers \& Mathematics with Applications  31(4-5),  121--135 (1996)

\bibitem{mehrmann2005balanced}
Mehrmann, V., Stykel, T.: Balanced truncation model reduction for large-scale
  systems in descriptor form. In: Dimension Reduction of Large-Scale Systems,
  pp. 83--115. Springer (2005)

\bibitem{schon2003modeling}
Schon, T., Gerdin, M., Glad, T., Gustafsson, F.: A modeling and filtering
  framework for linear differential-algebraic equations. In: Decision and
  Control, 2003. Proceedings. 42nd IEEE Conference on. vol.~1, pp. 892--897.
  IEEE (2003)

\bibitem{tran2017order}
Tran, H.D., Nguyen, L.V., Xiang, W., Johnson, T.T.: Order-reduction
  abstractions for safety verification of high-dimensional linear systems.
  Discrete Event Dynamic Systems  27(2),  443--461 (2017)

\end{thebibliography}



\appendix
\section{Proofs for  Section \ref{sec:decoupling}}
Before going forward, we present some useful properties of the matrix chain defined in Equation (\ref{eq:matrix_chain}) and the admissible projectors in Definition \ref{def:admissible_projectors} that are used in the decoupling process.

\begin{proposition}[Matrix chain properties]\label{pro:matrix_chain}
The matrix chain defined in Equation (\ref{eq:matrix_chain}) has the following properties:
\begin{equation}
\begin{split}
&E_{j + 1}P_j = E_j,~E_{j+1}Q_j = - A_jQ_j,~j = 0, 1, \cdots, \mu - 1, \\
&A_{\mu} = A_0 + E_{1}Q_0 + E_2Q_1 + \cdots + E_{\mu}Q_{\mu - 1} \\
&~~~~ = A_0 + E_{\mu}(P_{\mu - 1}\cdots P_1Q_0 + P_{\mu-2}\cdots P_2Q_1 + \cdots + Q_{\mu - 1}).
\end{split}
\end{equation}
\end{proposition}
\begin{proof}
From the definition of the matrix chain, we have: $E_{j+1}P_j = E_jP_j - A_jQ_jP_j = E_j(I_n - Q_j) = E_j$ and $E_{j+1}Q_j = E_iQ_i - A_iQ_i^2 = -A_iQ_i$.

From the definition and the first property, we have $A_{j + 1} = A_jP_j = A_j(I_n - Q_j) = A_j + E_{j+1}Q_j$. Therefore, $A_{\mu} = A_{\mu - 1} + E_{\mu}Q_{\mu-1} = \cdots = A_0 + E_{1}Q_0 + E_2Q_1 + \cdots + E_{\mu}Q_{\mu - 1}$. Further applying $E_{j+1}P_j = E_j$ completes the proof.
\end{proof}
\begin{proposition}[Admissible projectors properties]\label{pro:admissible_projectors}
The admissible projectors defined in Definition \ref{def:admissible_projectors} have the following properties:
\begin{equation}
P_jQ_i = Q_i,~ Q_jP_i = Q_j,~P_iP_jP_i = P_iP_j,~P_jP_iP_j = P_iP_j,~\forall j > i.
\end{equation}
\end{proposition}
\begin{proof}
From the definition of admissible projectors, we have: $P_jQ_i = (I_n - Q_j)Q_i = Q_i$, $Q_jP_i = Q_j(I_n - Q_i) = Q_j$, $P_iP_jP_i = P_iP_j(I_n - Q_i) = P_iP_j - P_iP_jQ_i = P_iP_j - P_iQ_i = P_iP_j$ since $P_iQ_i = 0, P_jQ_i = Q_i$. In addition, $P_jP_iP_j = (I_n - Q_j)P_iP_j = (P_i - Q_jP_i)P_j = (P_i - Q_j)P_j = P_iP_j - Q_jP_j = P_iP_j$ since $Q_jP_i = Q_j$ and $Q_jP_j = 0$.
\end{proof}
\subsection{Proof for Lemma \ref{lm:index-1-decoupling}}\label{proof:index-1}
\begin{proof}
Using the matrix chain defined in Equation (\ref{eq:matrix_chain}) and Proposition \ref{pro:matrix_chain}, we have: $E_0\dot{x}(t) = A_0x(t) + Bu(t) \rightarrow E_1P_0\dot{x}(t) = (A_1 - E_1Q_0)x(t) + Bu(t)$. Since the system is index-$1$, $E_1$ is non-singular. Therefore, we have:
\begin{equation}\label{eq:index-1-decouple-1}
P_0\dot{x}(t) + Q_0x(t) = E_1^{-1}[A_1x(t) + Bu(t)].
\end{equation}

By left multiplying Equation (\ref{eq:index-1-decouple-1}) by $P_0$ and $Q_0$ and using the fact that $A_1x(t) = A_1(P_0 + Q_0)x(t) = A_0P_0 (P_0 + Q_0)x(t) = A_0P_0x(t)$, the index-1 DAE system can be decoupled by:
\begin{equation*}
\begin{split}
&\dot{x}_1(t) = N_1x_1(t) + M_1u(t), \\
&x_2(t) = N_2x_1(t) + M_2u(t), \\
&x(t) = x_1(t) + x_2(t),
\end{split}
\end{equation*}
where $x_1(t) = P_0x(t)$, $N_1 = P_0E_1^{-1}A_0$, $M_1 = P_0E_1^{-1}B$; and $x_2(t) = Q_0x(t)$, $N_2 = Q_0E_1^{-1}A_0$, $M_2 = Q_0E_1^{-1}B$. This completes the proof.
\end{proof}

\subsection{Proof for Lemma \ref{lm:index-2-decoupling}}\label{proof:index-2}
\begin{proof}
Similar to the index-1 DAE case, using the matrix chain and Proposition \ref{pro:matrix_chain} we have: $E_2P_1P_0\dot{x}(t) = [A_2 - E_2(P_1Q_0 + Q_1)]x(t) + Bu(t)$. Further assume that $Q_0, Q_1$ are admissible projectors, then using Proposition \ref{pro:admissible_projectors} and the fact that $E_2$ is nonsingular we have:
\begin{equation}\label{eq:index-2-decouple-1}
 P_1P_0\dot{x}(t) + Q_0x(t) + Q_1x(t) = E_2^{-1}A_2x(t) + E_2^{-1}Bu(t).
\end{equation}
\end{proof}

Left multiplying Equation (\ref{eq:index-2-decouple-1}) by $P_0P_1$ leads to:
\begin{equation*}
P_0P_1^2P_0\dot{x}(t) + P_0P_1Q_0x(t) + P_0P_1Q_1x(t) = P_0P_1E_2^{-1}[A_2x(t) + Bu(t)].
\end{equation*}

From Proposition \ref{pro:admissible_projectors}, we have: $P_0P_1^2P_0 = P_0P_1P_0 = P_0P_1$, $P_0P_1Q_0 = P_0Q_0 = 0$, $P_0P_1Q_1 = 0$. In addition, $A_2x(t) = A_2[P_0P_1 + P_0Q_1 + Q_0]x(t) = A_2P_0P_1x(t)$ because $A_2P_0Q_1 = A_1P_1P_0Q_1 = A_0P_0P_1P_0Q_1 = A_0P_0P_1Q_1 = 0$ and $A_2Q_0 = A_1P_1Q_0 = A_1Q_0 = A_0P_0Q_0 = 0$. By combining these identities, the ODE subsystem can be derived by:
\begin{equation*}
\Delta_1:~~\dot{x}_1(t) = N_1x_1(t) + M_1u(t),
\end{equation*}
where $x_1(t) = P_0P_1x(t)$, $N_1 = P_0P_1E_2^{-1}A_2$ and $M_1 = P_0P_1E_2^{-1}B$.

Left multiplying Equation (\ref{eq:index-2-decouple-1}) by $P_0Q_1$ leads to:
\begin{equation*}
P_0Q_1P_1P_0\dot{x}(t) + P_0Q_1Q_0x(t) + P_0Q_1^2x(t) = P_0Q_1E_2^{-1}[A_2x_1(t) + Bu(t)].
\end{equation*}
Due to $P_0Q_1P_1P_0 = 0$, $P_0Q_1Q_0 = 0$ and $P_0Q_1^2 = P_0Q_1$, the first AC subsystem can be derived by:
\begin{equation*}
\Delta_2:~~x_2(t) = N_2x_1(t) + M_2u(t),
\end{equation*}
where $x_2(t) = P_0Q_1x(t)$, $N_2 = P_0Q_1E_2^{-1}A_2$ and $M_2 = P_0Q_1E_2^{-1}B$.

Left multiplying Equation (\ref{eq:index-2-decouple-1}) by $Q_0P_1$ leads to:
\begin{equation*}
Q_0P_1^2P_0\dot{x}(t) + Q_0P_1Q_0x(t) + Q_0P_1Q_1x(t) = Q_0P_1E_2^{-1}[A_2x(t) + Bu(t)],
\end{equation*}

Note that $Q_0P_1Q_0 = Q_0^2 = Q_0$, $Q_0P_1Q_1 = 0$ and $Q_0P_1^2P_0\dot{x}(t) = Q_0P_1P_0(P_0P_1 + P_0Q_1 + Q_0)\dot{x}(t) = Q_0P_1P_0Q_1\dot{x}(t) = Q_0[I_n - Q_1]P_0Q_1\dot{x}(t) = -Q_0Q_1P_0Q_1\dot{x}(t) = -Q_0Q_1\dot{x}_2(t)$. Therefore, the second AC subsystem can be derived below:
\begin{equation*}
\Delta_3:~~x_3(t) = N_3x_1(t) + M_3u(t) + L_3\dot{x}_2(t),
\end{equation*}
where $x_3(t) = Q_0x(t)$, $N_3 = Q_0P_1E_2^{-1}A_2$, $M_3 = Q_0P_1E_2^{-1}B$ and $L_3 = Q_0Q_1$.

It is easy to see that $I_n = P_0 + Q_0 = P_0(P_1 + Q_1) + P_0 = P_0P_1 + P_0Q_1 + Q_0$, therefore we have $x(t) = x_1(t) + x_2(t) + x_3(t)$. This completes the proof.

\subsection{Proof for Lemma \ref{lm:index-3-decoupling}}\label{proof:index-3}
Similar to the index-2 DAE case, using the matrix chain with admissible projectors leads to:
\begin{equation*}
E_3P_2P_1P_0\dot{x}(t) = [A_3 - E_3(P_2P_1Q_0 + P_2Q_1 + Q_2)]x(t) + Bu(t),
\end{equation*}
or equivalently,
\begin{equation}\label{eq:index-3-decouple-1}
P_2P_1P_0\dot{x}(t) + Q_0x(t) + Q_1x(t) + Q_2x(t) = E_3^{-1}[A_3x(t) + Bu(t)].
\end{equation}
By left multiplying Equation (\ref{eq:index-3-decouple-1}) by $P_0P_1P_2$ with noticing that $P_0P_1P_2P_2P_1P_0 = P_0P_1P_2P_1P_0 = P_0P_1P_2P_0 = P_0P_1P_2(I_n - Q_0) = P_0P_1P_2$, $P_0P_1P_2Q_0 = P_0P_1Q_0 = P_0Q_0 = 0$, $P_0P_1P_2Q_1 = P_0P_1Q_1 = 0$, $P_0P_1P_2Q_2 = 0$ and $A_3x(t) = A_3(P_0P_1P_2 + Q_0P_1P_2 + Q_1P_2 + Q_2)x(t) = A_3P_0P_1P_2x(t)$, the ODE subsystem of the DAE system can be derived by:
\begin{equation*}
\Delta_1:~~x_1(t) = N_1x_1(t) + M_1u(t),
\end{equation*}
where $x_1(t) = P_0P_1P_2x(t)$, $N_1 = P_0P_1P_2E_3^{-1}A_3$ and $M_1 = P_0P_1P_2E_3^{-1}B$.

Similarly, by left multiplying Equation (\ref{eq:index-3-decouple-1}) by $P_0P_1Q_2$, the first AC subsystem of the DAE system can be derived below:
\begin{equation*}
\Delta_2:~~x_2(t) = N_2x_1(t) + M_2u(t),
\end{equation*}
where $x_2(t) = P_0P_1Q_2x(t)$, $N_2 = P_0P_1Q_2E_3^{-1}A_3$ and $M_2 = P_0P_1Q_2E_3^{-1}B$.

Left multiplying Equation (\ref{eq:index-3-decouple-1}) by $P_0Q_1P_2$ yields:
\begin{equation*}
P_0Q_1P_2P_1P_0\dot{x}(t) + P_0Q_1x(t) = P_0Q_1P_2E_3^{-1}[A_3x(t) + Bu(t)],
\end{equation*}
in which: $P_0Q_1P_2P_1P_0\dot{x}(t) = P_0Q_1(I_n - Q_2)P_1P_0\dot{x}(t) = -P_0Q_1Q_2\dot{x}(t) = -P_0Q_1Q_2(P_0P_1P_2 + P_0P_1Q_2)\dot{x}(t) = -P_0Q_1Q_2\dot{x}_2(t)$. Therefore, we have the second AC subsystem as follows.
\begin{equation*}
\Delta_3:~~x_3(t) = N_3x_1(t) + M_3u(t) + L_3\dot{x}_2(t),
\end{equation*}
where $x_3(t) = P_0Q_1x(t)$, $N_3 = P_0Q_1P_2E_3^{-1}A_3$, $M_3 = P_0Q_1P_2E_3^{-1}B$ and $L_3 = P_0Q_1Q_2$.

Left multiplying Equation (\ref{eq:index-3-decouple-1}) by $Q_0P_1P_2$ yields:
\begin{equation*}
Q_0P_1P_2^2P_1P_0\dot{x}(t) + Q_0x(t) = Q_0P_1P_2E_3^{-1}[A_3x(t) + Bu(t)],
\end{equation*}
in which: $Q_0P_1P_2^2P_1P_0\dot{x}(t) = Q_0P_1P_2P_0\dot{x}(t) = - (Q_0Q_1 + Q_0P_1Q_2)\dot{x}(t) = -(Q_0Q_1 + Q_0P_1Q_2)(P_0P_1P_2 + P_0P_1Q_2 + P_0Q_1 + Q_0)\dot{x}(t) = -Q_0P_1Q_2\dot{x}_2(t) - Q_0Q_1\dot{x}_3(t)$. Therefore, the last AC subsystem of the DAE system can be derived below:
\begin{equation*}
\Delta_4:~~x_4(t) = N_4x_1(t) + M_4u(t) + L_4\dot{x}_3(t) + Z_4\dot{x}_2(t),
\end{equation*}
where $x_4(t) = Q_0x(t)$, $N_4 = Q_0P_1P_2E_3^{-1}A_3$, $M_4 = Q_0P_1P_2E_3^{-1}B$, $L_4 = Q_0Q_1$ and $Z_4 = Q_0P_1Q_2$.

It is easy to see that $x(t) = (P_0P_1P_2 + P_0P_1Q_2 + P_0Q_1 + Q_0)x(t) = x_1(t) + x_2(t) + x_3(t) + x_4(t)$. This completes the proof.

\subsection{Proof for Proposition \ref{pro:orthogonal-projector}}\label{proof:orthogonal-projector}
\begin{proof}
Note that $K = [K_1~K_2]$ is an unitary matrix, i.e., $KK^T = K^TK = I_n$. Consequently, $K_1^TK_2 = 0$ and $K_2^TK_2 = I_{n-r}$. It is easy to see that $Z = L_1SK_1^T$, so we have $ZQ = 0$. In addition, $Q = Q^T$ and $Q^2 = K_2K_2^TK_2K_2^T = K_2K_2^T = Q$. Therefore, $Q$ is a orthogonal projector on $Z$. This completes the proof.
\end{proof}
\subsection{Proof for Lemma \ref{lm:admissible-projectors-index-2}}\label{proof:admissible-projectors-index-2}
\begin{proof}
We need to prove that $Q_1^*$ is also a projector on $E_1$ and $Q_1^*Q_0^* = 0$. We have $Q_1^*Q_1 = -Q_1E_2^{-1}A_1Q_1 = Q_1E_2^{-1}E_2Q_1 = Q_1^2 = Q_1$ since $A_1Q_1 = -E_2Q_1$ (Proposition \ref{pro:matrix_chain}). We also have $Q_1Q_1^* = -Q_1^2E_2^{-1}A_1 = -Q_1E_2^{-1}A_1 = Q_1^*$. Therefore, $Q_1^*$ is also a projector on $E_1$. In addition, $Q_1^*Q_0^* = -Q_1E_2^{-1}A_0P_0Q_0 = 0$. This completes the proof.
\end{proof}
\subsection{Proof for Lemma \ref{lm:admissible-projectors-index-3}}\label{proof:admissible-projectors-index-3}
We first need to prove that $Q_1^*$ and $Q_2^*$ are respectively projectors on $E_1$ and $E_2^{\prime}$. Then we prove that $Q_0^*$, $Q_1^*$ and $Q_2^*$ satisfy admissible conditions. We have $Q_2^{\prime}Q_2 = Q_2E_3^{-1}E_3Q_2 = Q_2^2 = Q_2$ and $Q_2Q_2^{\prime} = -Q_2^2E_3^{-1}A_2 = -Q_2E_3^{-1}A_2 = Q_2^{\prime}$. Thus, $Q_2^{\prime}$ is a projector on $E_2$. Consequently, we can check that $P_2^{\prime}P_2 = (I_n - Q_2^{\prime})(I_n - Q_2) = P_2^{\prime}$.

One can see that $Q_2^{\prime}Q_1 = 0$ due to $A_2Q_1 = A_1P_1Q_1 = 0$, consequently, $P_2^{\prime}Q_1 = (I_n - Q_2^{\prime})Q_1 = Q_1$. Therefore, $Q_1^{\prime}Q_1 = Q_1P_2^{\prime}E_3^{-1}E_2Q_1 = Q_1P_2^{\prime}E_3^{-1}E_3P_2Q_1 = Q_1P_2^{\prime}P_2Q_1$. In addition, we have $P_2^{\prime}P_2 = P_2^{\prime}$. Consequently, $Q_1^{\prime}Q_1 = Q_1^2 = Q_1$. Furthermore, it is easy to see $Q_1Q_1^{\prime} = Q_1^{\prime}$. Therefore, $Q_1^{\prime}$ is a projector on $E_1$. Similarly, it is easy to check that $Q_2^*$ is also a projector on $E_2^{\prime}$ since $Q_2^{*}Q_2^{\prime\prime} = Q_2^{\prime\prime}$ and $Q_2^{\prime\prime}Q_2^* = Q_2^*$.

Next, we prove that $Q_1^*Q_0^* = 0$, $Q_2^*Q_0^*$ and $Q_2^*Q_1^* = 0$. We have $Q_1^*Q_0^* = -Q_1P_2^{\prime}E_3^{-1}A_1Q_0 = 0$ due to $A_1Q_0 = A_0P_0Q_0 = 0$; $Q_2^*Q_1^* = 0$ because of $A_2^{\prime}Q_1^{\prime} = A_1P_1^{\prime}Q_1^{\prime} = 0$; $Q_2^*Q_0^* = 0$ because of $A_2^{\prime}Q_0 = A_1P_1^{\prime}Q_0 = A_1Q_0 = A_0P_0Q_0 = 0$ (note that $Q_1^{\prime}Q_0 = 0 \rightarrow P_1^{\prime}Q_0 = Q_0$). This completes the proof.

\section{Proof for  Section \ref{sec:reachability}}
\subsection{Proof for Lemma \ref{lm:reachable-set}}\label{proof:reachable-set}
\begin{proof}
Let $x_1(t) \in \Theta_1(t)$ is a solution of the ODE subsystem at time $t$. Then, we have: 1) if the autonomous DAE system is index-1, from Lemma \ref{lm:index-1-decoupling}, the solution of the DAE system is $x(t) = x_1(t) + x_2(t) = (I_n + N_2)x_2(t)$, consequently, the reachable set of the autonomous DAE system at time $t$ is $\Theta(t) = \langle (I_n + N_2)V_1(t), P \rangle$; 2) if the DAE system is index-2, from Lemma \ref{lm:index-2-decoupling}, the solution of the DAE system is $x(t) = x_1(t) + x_2(t) + x_3(t) = x_1(t) + N_2x_1(t) + N_3x_1(t) + L_3\dot{x}_2(t) = (I_n + N_2 + N_3 + L_3N_2N_1)x_1(t)$, consequently, the reachable set of the DAE system at time $t$ is $\Theta(t) = \langle (I_n + N_2 + N_3 + L_3N_2N_1)V_1(t), P \rangle$; 3) if the DAE system is index-3, from Lemma \ref{lm:index-3-decoupling}, the solution of the DAE system is $x(t) = x_1(t) + x_2(t) + x_3(t) + x_4(t) = x_1(t) + N_2x_1(t) + N_3x_1(t) + L_3\dot{x}_2(t) + N_4x_1(t) + L_4\dot{x}_3(t) + Z_4\dot{x}_2(t) =  x_1(t) + N_2x_1(t) + N_3x_1(t) + L_3N_2N_1x_1(t) + N_4x_1(t) + L_4[N_3N_1x_1(t) + L_3N_2N_1^2x_1(t)] + Z_4N_2N_1x_1(t) = (I_n + N_2 + N_3 + L_3N_2N_1 + N_4 + L_4N_3N_1 + L_4L_3N_2N_1 + Z_4N_2N_1)x_1(t)$, consequently, the reachable set of the DAE system at time $t$ is: $\Theta(t) = \langle (I_n + N_2 + N_3 + N_4 + L_3N_2N_1 + L_4N_3N_1 + L_4L_3N_2N_1^2 + Z_4N_2N_1)V_1(t), P \rangle$. This completes the proof.
\end{proof}

\section{Safety verification and falsification of the interconnected rotating mass system}\label{description:ex1}
\begin{example}[Interconnected rotating masses \cite{schon2003modeling}]\label{ex1}
This is an index-2 DAE system with four state variables $x(t) = [z^T_1(t), z_2^T(t), M_2^T(t), M_3^T(t)]^T$ and two inputs $u(t) = [M_1(t)^T, M_4^T(t)]^T$ where $z_1(t)$ and $z_2(t)$ are the angular velocities of the first and the second masses respectively, and $M_2(t)$ and $ M_3(t)$ are the torques on the connection of these two masses. $M_1(t)$ and $M_4(t)$ are the input torques applied to the first and the second masses. The system matrices $E$, $A$, and $B$ are described by:

\begin{equation*}
E = \begin{bmatrix} J_1 & 0 & 0 & 0 \\ 0 & J_2 & 0 & 0 \\ 0 & 0 & 0 & 0 \\ 0 & 0 & 0 & 0 \\ \end{bmatrix},~ A = \begin{bmatrix} 0 & 0 & 1 & 0 \\ 0 & 0 & 0 & 1 \\ 0 & 0 & -1 & -1\\ -1 & 1 & 0 & 0 \\\end{bmatrix},~B = \begin{bmatrix} 1 & 0 \\ 0 & 1 \\ 0 & 0 \\ 0 & 0 \\ \end{bmatrix},~J_1 = 1, J_2 = 2.
\end{equation*}

We are interested in the angular velocity, $z_1(t)$, and the torque, $M_2(t)$, of the first mass. The unsafe set for the system is defined by: $Unsafe \triangleq M_2(t) \leq -0.9$. The system is controlled by \emph{sine} function inputs defined as follows.

\begin{equation*}
\begin{bmatrix} \dot{M}_1(t) \\ \dot{M}_4(t) \\ \end{bmatrix} = \begin{bmatrix} 0 & 1 \\ -1 & 0 \\\end{bmatrix}\begin{bmatrix} M_1(t) \\ M_4(t) \\ \end{bmatrix}, u(0) = \begin{bmatrix} M_1(0) \\ M_4(0) \\\end{bmatrix} \in U.
\end{equation*}

We transform the system with given inputs to an autonomous DAE system as described in Equation (\ref{eq:auto-dae}). A consistent initial set of states $\Theta(0) = \langle V(0), P\rangle$ for the autonomous DAE system is chosen below.

\begin{equation*}
V(0) = \begin{bmatrix}0 & 0 \\ 0 & 0 \\ 0.513 & 0 \\ -0.513 & 0 \\ -0.616 & 0.447 \\ 0.308 & 0.894 \\ \end{bmatrix},~P(\alpha) \triangleq C\alpha \leq d,~C = \begin{bmatrix} 1 & 0 \\ -1 & 0 \\ 0 & 1 \\ 0 & -1 \\\end{bmatrix},~d = \begin{bmatrix} 0.2 \\ -0.1 \\ 1.2 \\ -1.0 \\ \end{bmatrix}
\end{equation*}

Using Algorithm \ref{alg:admissible-projectors}, we construct admissible projectors $Q_0$, $Q_1$ for the autonomous DAE system below.
\begin{equation*}
Q_0 = \begin{bmatrix} 0 & 0 & 0 & 0 & 0 & 0 \\ 0 & 0 & 0 & 0 & 0 & 0 \\ 0 & 0 & 1 & 0 & 0 & 0 \\ 0 & 0 & 0 & 1 & 0 & 0 \\ 0 & 0 & 0 & 0 & 0 & 0 \\ 0 & 0 & 0 & 0 & 0 & 0 \\ \end{bmatrix}, ~ Q_1 = \begin{bmatrix} \frac{2}{3} & \frac{-2}{3} & 0 & 0 & 0 & 0 \\ \frac{-1}{3} & \frac{1}{3} & 0 & 0 & 0 & 0 \\ \frac{2}{3} & \frac{-1}{3} & 0 & 0 & 0 & 0 \\ \frac{-2}{3} & \frac{2}{3} & 0 & 0 & 0 & 0 \\ 0 & 0 & 0 & 0 & 0 & 0 \\ 0 & 0 & 0 & 0 & 0 & 0 \\ \end{bmatrix}.
\end{equation*}

Using these admissible projectors and Lemma \ref{lm:index-2-decoupling}, the autonomous DAE system can be decoupled into an equivalent decoupled system with the following matrix coefficients:
\begin{equation*}
N_1 = \begin{bmatrix} 0 & 0 & 0 & 0 & \frac{1}{3} & \frac{1}{3} \\ 0 & 0 & 0 & 0 & \frac{1}{3} & \frac{1}{3} \\ 0 & 0 & 0 & 0 & 0 & 0 \\ 0 & 0 & 0 & 0 & 0 & 0 \\ 0 & 0 & 0 & 0 & 0 & 1 \\ 0 & 0 & 0 & 0 & -1 & 0 \\ \end{bmatrix},~N_2 = 0,~N_3 = \begin{bmatrix} 0 & 0 & 0 & 0 & 0 & 0 \\ 0 & 0 & 0 & 0 & 0 & 0 \\ 0 & 0 & 0 & 0 & \frac{-2}{3} & \frac{1}{3} \\ 0 & 0 & 0 & 0 & \frac{2}{3} & \frac{-1}{3} \\ 0 & 0 & 0 & 0 & 0 & 0 \\ 0 & 0 & 0 & 0 & 0 & 0 \\ \end{bmatrix},~L_3 = \begin{bmatrix} 0 & 0 & 0 & 0 & 0 & 0 \\ 0 & 0 & 0 & 0 & 0 & 0 \\ \frac{2}{3} & \frac{-2}{3} & 0 & 0 & 0 & 0 \\ \frac{-2}{3} & \frac{2}{3} & 0 & 0 & 0 & 0 \\ 0 & 0 & 0 & 0 & 0 & 0 \\ 0 & 0 & 0 & 0 & 0 & 0 \\ \end{bmatrix}.
\end{equation*}

Using the decoupled system, we verify the safety property of the original DAE system over time up to $T = 10$ seconds using fixed time step $h = 0.01$. In $0.37$ second, our approach can show that the system is simulationally unsafe and produces a trace violating the safety property which is depicted in Figure \ref{fig:ex1-unsafe-trace}. Figure \ref{fig:ex1-reachset} shows the reachable set of the outputs $z_1(t)$ and $M_2(t)$ of the DAE system.%
\begin{figure}%
\centering
\subfigure[Unsafe trace]{%
\label{fig:ex1-unsafe-trace}%
\includegraphics[width=0.48\textwidth]{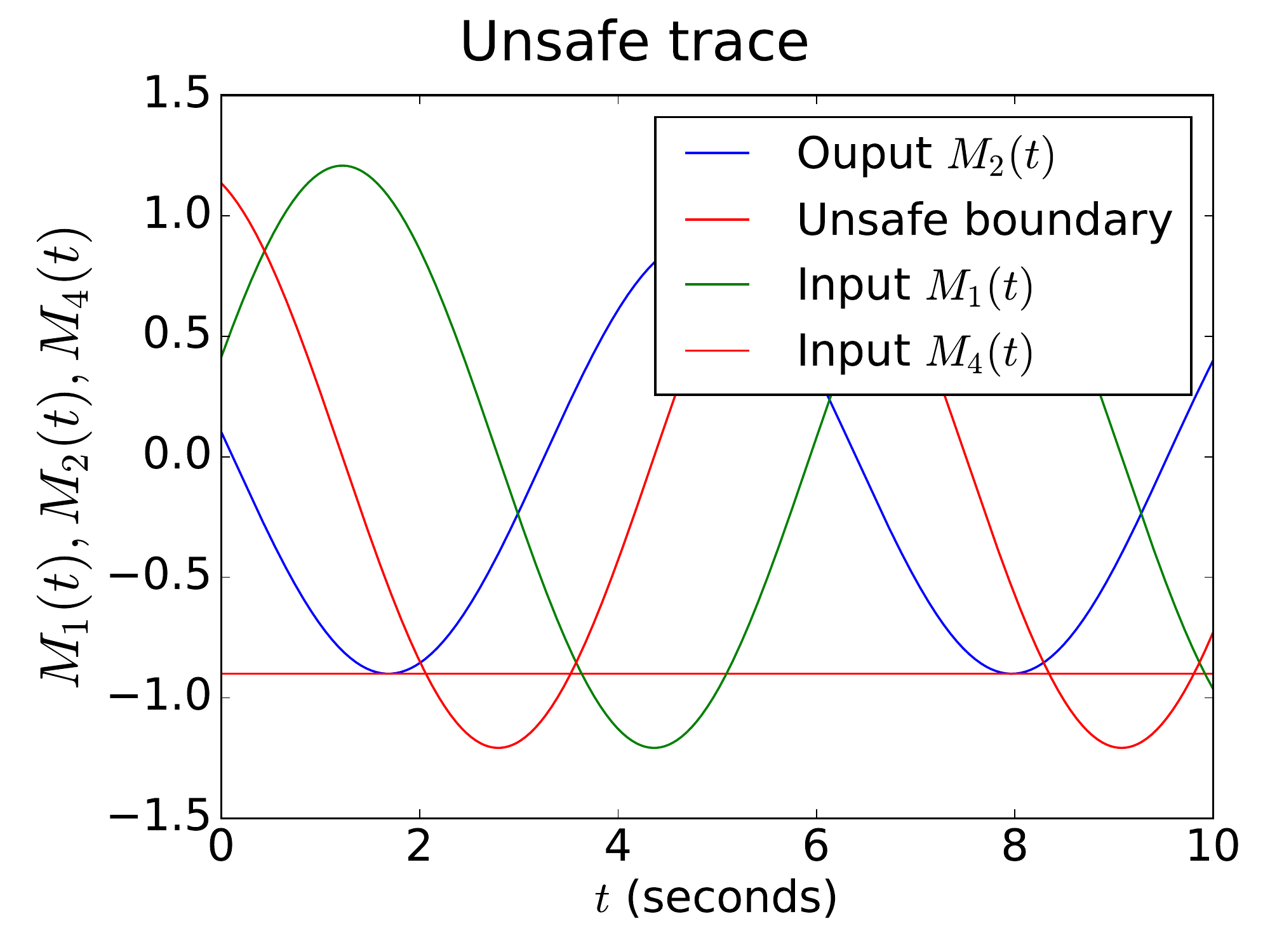}}%
\subfigure[Reachable set]{%
\label{fig:ex1-reachset}%
\includegraphics[width=0.48\textwidth]{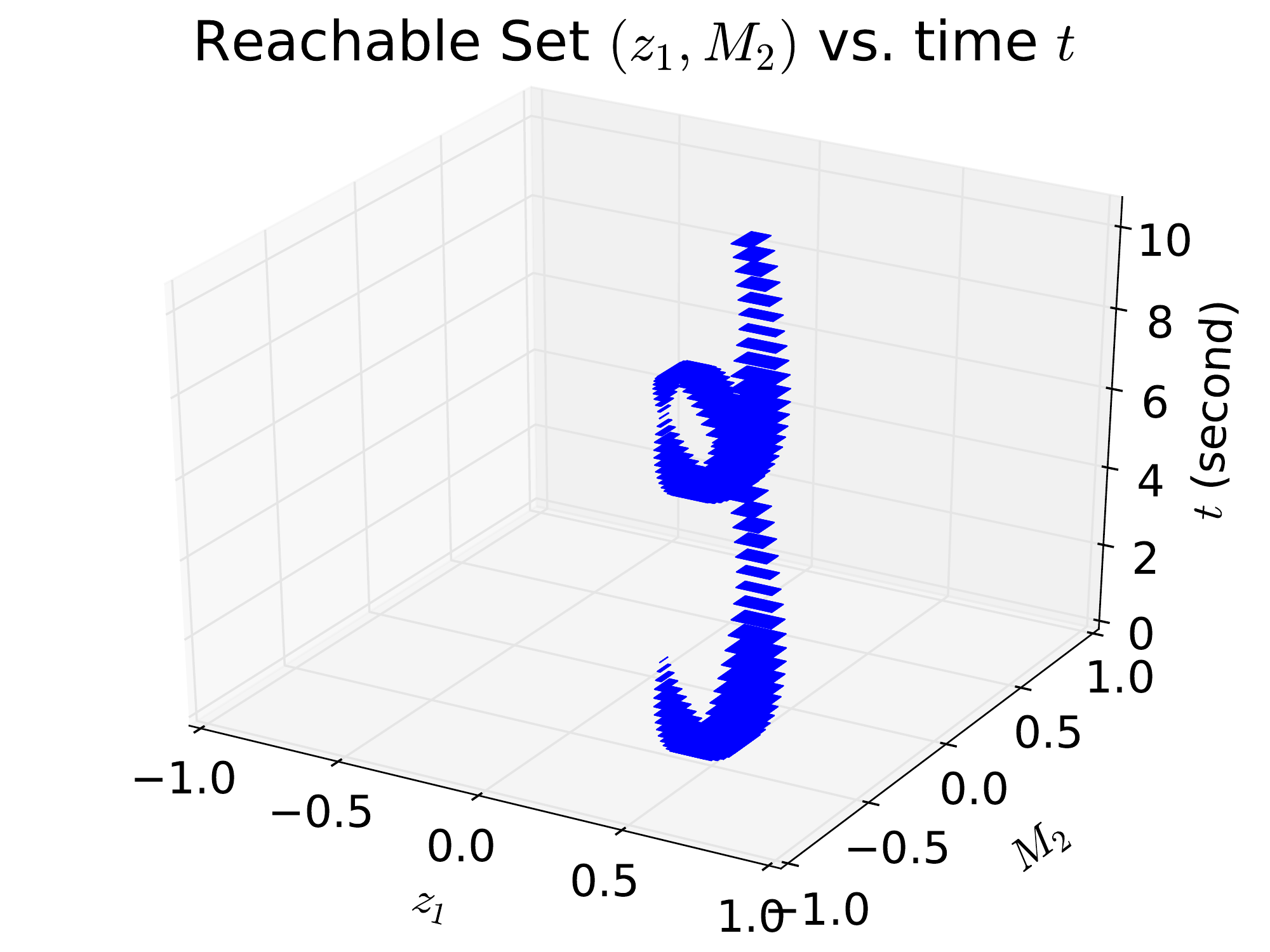}}%
\caption{Unsafe trace and reachable set of Example \ref{ex1}}
\end{figure}
We have shown in detail how our approach works. Next, we discuss the time performance of our approach by analyzing the safety of the semidiscretized Stokes Equation benchmark.

\end{example}

\section{Mathematical model of the Stokes-equation}\label{description:ex2}
The considered Stokes-equation is given as follows.

\begin{equation}
\begin{split}
&\frac{\partial{v}}{\partial{t}} = \Delta v - \nabla \rho + f, ~~~\text{in}~\Omega \times (0, T), \\
&\nabla v = 0,~~~ \text{in}~\Omega \times (0, T),
\end{split}
\end{equation}
where $v(\zeta, t) \in \mathbb{R}^2$ is the velocity vector, $\rho(\zeta, t) \in \mathbb{E}$ is the pressure, $f(\zeta, t) \in \mathbb{R}^2$ is the vector of external forces, $\Omega = (0,1)^2 \subset \mathbb{R}^2$ is a square domain, $T$ is the endpoint of the time interval, $\nabla$ denotes \emph{the divergence operator} and $\Delta = \nabla^2$.

In this paper, we use Dirichlet boundary condition for the Stokes equation. This condition means that the velocity equals zero on the boundary of the domain. Semi-discretizing the Stokes equation using the well-known MAC scheme [xx] (a scheme leveraging finite element method) leads to an index-2 DAE of the form (\ref{eq:dae}) with the following system matrices:
\begin{equation}
E = \begin{bmatrix} I_{n_v} & 0 \\ 0 & 0 \\\end{bmatrix},~A = \begin{bmatrix} A_{11} & A_{12}\\ A_{12}^T & 0 \\\end{bmatrix},~B = \begin{bmatrix} B_1\\ B_2\end{bmatrix},~x = \begin{bmatrix} v_h \\ \rho_h \\\end{bmatrix},
\end{equation}
where $v_h \in \mathbb{R}^{n_v}$ and $\rho_h \in \mathbb{R}^{n_{\rho}}$ are the semi-discretized vectors of velocity and pressure, $A_{11} \in \mathbb{R}^{n_v \times n_v}$ is the discrete Laplace operator, $A_{12} \in \mathbb{R}^{n_v \times n_{\rho}}$ and $A_{12}^T \in \mathbb{R}^{n_{\rho} \times n_v}$ are the discrete gradient and divergence operators respectively.


%






\end{document}